\documentclass[a4paper]{lipics-v2019}

\nolinenumbers
\hideLIPIcs

\title{New Partitioning Techniques and Faster Algorithms for Approximate Interval Scheduling} %

\titlerunning{Faster Approximate Interval Scheduling} %

 \author{Spencer Compton}{Stanford University}{comptons@stanford.edu}{}{}
 
 \author{Slobodan Mitrovi\'c}{UC Davis}{smitrovic@ucdavis.edu}{}{}

 \author{Ronitt Rubinfeld}{MIT, CSAIL}{ronitt@csail.mit.edu}{}{}%

 \authorrunning{S. Compton, S. Mitrovi\'c, R. Rubinfeld} %


\usepackage[utf8]{inputenc}
\usepackage[rflt]{floatflt}
\usepackage[T1]{fontenc}
\usepackage{graphicx}
\usepackage{epsfig}
\usepackage{color}
\usepackage{verbatim}
\usepackage{amsthm, amssymb}
\usepackage{amsmath}
\usepackage{thmtools,thm-restate}
\usepackage{wrapfig}
\usepackage{xspace}

\usepackage[disable]{todonotes}

\definecolor{darkgreen}{rgb}{0,0.5,0}
\usepackage{hyperref}
\hypersetup{
    unicode=false,          %
    colorlinks=true,        %
    linkcolor=blue,          %
    citecolor=darkgreen,        %
    filecolor=magenta,      %
    urlcolor=cyan           %
}

\usepackage[linesnumbered,vlined,ruled]{algorithm2e}
	\DontPrintSemicolon
  \SetKwInOut{Input}{Input}
  \SetKwInOut{Output}{Output}
  
  \SetCommentSty{mycommfont}

\usepackage[disableredefinitions,roman]{complexity}
\usepackage[framemethod=tikz]{mdframed}
\global\mdfdefinestyle{myframe}{leftmargin=.75in,rightmargin=.75in,linecolor=black,linewidth=1.5pt,innertopmargin=10pt,innerbottommargin=10pt} 
\usepackage{mdframed}
\usepackage{framed}
\usepackage{nicefrac}

\newtheorem{invariant}{Invariant}

\usepackage[capitalize, nameinlink]{cleveref}
\crefname{theorem}{Theorem}{Theorems}
\Crefname{lemma}{Lemma}{Lemmas}
\Crefname{alg}{Algorithm}{Algorithms}
\Crefname{claim}{Claim}{Claims}
\Crefname{observation}{Observation}{Observations}
\Crefname{invariant}{Invariant}{Invariants}
\Crefname{algorithm}{Algorithm}{Algorithms}

\newcommand{\rb}[1]{\left( #1 \right)}

\newcommand{\cJ}{\mathcal{J}}
\newcommand{\cE}{\mathcal{E}}
\newcommand{\cR}{\mathcal{R}}

\newcommand{\eps}{\varepsilon}
\newcommand{\fcon}{f_{contain}}
\newcommand{\fint}{f_{intersect}}

\newcommand{\ProbeBased}{\textsc{Probe-based-Opt}\xspace}
\newcommand{\MaxIS}{\textsc{Maximum-IS}\xspace}
\newcommand{\CUR}{\textsc{CUR}\xspace}
\newcommand{\earliest}{\textsc{earliest}\xspace}
\newcommand{\key}[1]{\textsc{key}(#1)\xspace}

\begin{document}

\maketitle

\begin{abstract}
Interval scheduling is a basic problem in the theory of algorithms and a classical task in combinatorial optimization.
We develop a set of techniques for partitioning and grouping jobs based on their starting and ending times, that enable us to view an instance of interval scheduling on \emph{many} 
jobs as a union of 
multiple interval scheduling instances, each containing only \emph{a few} jobs. Instantiating these techniques in dynamic and local settings of computation leads to several new results.

For $(1+\eps)$-approximation of job scheduling of $n$ jobs on a single machine, we develop a fully dynamic algorithm with $O(\nicefrac{\log{n}}{\eps})$ update and $O(\log{n})$ query worst-case time. Further, we design a local computation algorithm that uses only $O(\nicefrac{\log{N}}{\eps})$ queries when all jobs are length at least $1$ and have starting/ending times within $[0,N]$.
Our techniques are also applicable in a setting where jobs have rewards/weights. For this case we design a fully dynamic \emph{deterministic} algorithm whose worst-case update and query time are $\poly(\log n,\frac{1}{\eps})$. 
Equivalently, this is \emph{the first} algorithm that maintains a $(1+\eps)$-approximation of the maximum independent set of a collection of weighted intervals in $\poly(\log n,\frac{1}{\eps})$ time updates/queries.
This is an exponential improvement in $1/\eps$ over the running time of a randomized algorithm of Henzinger, Neumann, and Wiese ~[SoCG, 2020], while also removing all dependence on the values of the jobs' starting/ending times and rewards, as well as removing the need for any randomness.

We also extend our approaches for interval scheduling on a single machine to examine the setting with $M$ machines.
\end{abstract}

\section{Introduction}

Job scheduling is a fundamental task in optimization, with applications ranging from resource management in computing~\cite{salot2013survey,sharma2010survey} to operating transportation systems~\cite{kolen2007interval}.
Given a collection of \emph{machines} and a set of \emph{jobs} (or tasks) to be processed, the goal of job scheduling is to assign those jobs to the machines while respecting certain constraints.
Constraints set on jobs may significantly vary. In some cases a job has to be scheduled, but the starting time of its processing is not pre-specified. In other scenarios a job can only be scheduled at a given time, but there is a flexibility on whether to process the job or not.
Frequent objectives for this task can include either maximizing the number of scheduled jobs or minimizing needed time to process all the given jobs.

An important variant of job scheduling is the task of \emph{interval scheduling}: here each job has a specified starting time and its length, but a job is not required to be scheduled. Given $M$ machines, the goal is to schedule as many jobs as possible. More generally, each job is also assigned a \emph{reward} or weight, which can be thought of as a payment received for processing the given job. If a job is not processed, the payment is zero, i.e., there is no penalty. We refer to this variant as \emph{weighted interval scheduling}.
This problem in a natural way captures real-life scenarios. For instance, consider an assignment of crew members to flights, where our goal is to assign (the minimum possible) crews to the specified flights. In the context of interval scheduling, flights can be seen as jobs and the crew members as machines~\cite{kolen2007interval,mingozzi1999set}.
Interval scheduling also has applications in geometrical tasks -- it can be see as a task of finding a collection of non-overlapping geometric objects. In this context, its prominent applications are in VLSI design~\cite{hochbaum1985approximation} and map labeling~\cite{agarwal1998label,verweij1999optimisation}.

The aforementioned scenarios are executed in different computational settings. For instance, some use-cases are dynamic in nature, e.g., a flight gets cancelled. Then, in certain cases we have to make online decisions, e.g., a customer must know immediately whether we are able to accept its request or not. While in some applications there might be so many requests that we would like to design extremely fast ways of deciding whether a given request/job can be scheduled or not, e.g., providing an immediate response to a user submitting a job for execution in a cloud.
In this work, our aim is to develop methods for interval scheduling that can be turned into efficient algorithms across many computational settings:
\begin{center}
    \emph{Can we design unified techniques for approximating interval scheduling very fast?}
\end{center}
In this paper we develop fast algorithms for the dynamic and local settings of computation. We also give a randomized black-box approach that reduces the task of interval scheduling on multiple machines to that of interval scheduling on a single machine by paying only $2 - 1/M$ in the approximation factor for unweighted jobs, where $M$ is the number of machines, and $e$ in approximation factor for weighted jobs.
A common theme in our algorithms is partitioning jobs over dimensions (time and machines). It is well studied in the dynamic setting how to partition the time dimension to enable fast updates. It is also studied how to partition over the machines to enable strong approximation ratios for multiple-machine scheduling problems. We design new partitioning methods for the time dimension (starting and ending times of jobs), introduce a partitioning method over machines, and examine the relationship of partitioning over the time dimension and machines simultaneously in order to solve scheduling problems. We hope that, in addition to improving the best-known results, our work provides a new level of simplicity and cohesiveness for this style of approach.

\subsection{Computation Models}
In our work, we focus on the following two models of computation.
\subparagraph{Dynamic setting.}
Our algorithms for the fully dynamic setting design data structures that maintain an approximately optimal solution to an instance of the interval scheduling problem while supporting insertions and deletions of jobs/intervals. The data structures also support queries of the maintained solution's total weight and whether or not a particular interval is used in the maintained solution.

\subparagraph{Local computation algorithms (LCA).}
The LCA model was introduced by Rubinfeld et al.~\cite{rubinfeld2011fast} and Alon et al.~\cite{alon2012space}.
In this setting, for a given job $J$ we would like to output whether $J$ is scheduled or not, but we do not have a direct access to the entire list of input jobs. Rather, the LCA is given  access to an oracle that returns answers to questions
of the  form: ``\emph{What is the input job with the earliest ending time among those jobs that start after time $x$?}''
The goal of the LCA in this setting is to provide (yes/no) answers to user queries that ask 
``Is job $i$ scheduled?" (and, if applicable, ``On which machine?''), in such a manner
that all answers should be consistent with the same valid solution, while using as few oracle-probes as possible.

\subsection{Our Results}

Our first result, given in \cref{section:dynamic-unit}, focuses on designing an efficient dynamic algorithm for unweighted interval scheduling on a single machine. Prior to our work, the state-of-the-art result for this problem was due to \cite{bhore2020dynamic}, who design an algorithm with $O(\nicefrac{\log{n}}{\eps^2})$ update and query time. We provide an improvement in the dependence on $\eps$.
\begin{restatable}[Unweighted dynamic, single machine]{theorem}{theoremunweightedM}
\label{theorem:unweighted-M=1}
Let $\cJ$ be a set of $n$ jobs.
For any $\eps > 0$, there exists a fully dynamic algorithm for $(1+\varepsilon)$-approximate unweighted interval scheduling for $\cJ$ on a single machine performing updates in $O\rb{\frac{\log(n)}{\varepsilon}}$ and queries in $O(\log(n))$ worst-case time. 
\end{restatable}

\cref{theorem:unweighted-M=1} can be seen as a warm-up for our most challenging and technically involved result, which is an algorithm for the dynamic  \emph{weighted} interval scheduling problem on a single machine. We present our approach in detail in \cref{section:dynamic-weighted}. As a function of $1/\eps$, our result constitutes an exponential improvement compared to the running times obtained in \cite{henzinger2020dynamic}. We also remove all use of randomness, remove all dependence on the job starting/ending times (previous work crucially used assumptions on the coordinates to bound the ratio of jobs' lengths by a parameter $N$), and remove all dependence on the value of the job rewards.
\begin{restatable}[Weighted dynamic, single machine]{theorem}{theoremweighteddynamic}
\label{theorem:weighted-dynamic-M=1}
Let $\cJ$ be a set of $n$ weighted jobs.
For any $\eps > 0$, there exists a fully dynamic algorithm for $(1+\varepsilon)$-approximate weighted interval scheduling for $\cJ$ on a single machine performing updates and queries in worst-case time $T \in \poly(\log n,\frac{1}{\eps})$. The exact complexity of $T$ is given by
\[
O\rb{\frac{\log^{12}(n)}{\eps^{7}} + \frac{\log^{13}(n)}{\eps^{6}}}.
\]
\end{restatable}

\subsubsection{Implications in Other Settings}

\subparagraph{Local Computation Algorithms.}
We show that the ideas we developed to obtain \cref{theorem:unweighted-M=1} can also be efficiently implemented in the local setting, as we explain in detail in \cref{section:local} and prove the following claim. This is the first non-trivial local computation algorithm for the interval scheduling problem. 
\begin{restatable}[Unweighted LCA, single machine]{theorem}{theoremunweightedMlocal}
\label{theorem:unweighted-M=1-local}
Let $\cJ$ be a set of $n$ jobs with length at least $1$ and ending times upper-bounded by $N$.
For any $\eps > 0$, there exists a local computation algorithm for $(1+\varepsilon)$-approximate  unweighted interval scheduling for $\cJ$ on a single machine using $O\rb{\frac{\log{N}}{\varepsilon}}$ probes.
\end{restatable}

\subparagraph{Multiple machines.}
By building on techniques we introduced to prove \cref{theorem:unweighted-M=1,theorem:unweighted-M=1-local}, we show similar results in \cref{section:scheduling} in the case of interval scheduling on multiple machines at the expense of slower updates. To the best of our knowledge, these results initiate a study of dynamic and local interval scheduling in the general setting, i.e., in the setting of maximizing the total reward of jobs scheduled on multiple machines.

\subsection{Related Work}
The closest prior work to ours is that of Henzinger et al.~\cite{henzinger2020dynamic} and of Bhore et al.~\cite{bhore2020dynamic}.  \cite{henzinger2020dynamic} studies $(1+\eps)$-approximate dynamic interval scheduling for one machine in both the weighted and unweighted setting. Unlike our main result in \cref{theorem:weighted-dynamic-M=1}, they assume jobs have rewards within $[1,W]$, assume jobs have length at least 1, and assume all jobs start/end within times $[0,N]$. They obtain randomized algorithms with $O(\exp(1/\eps) \log^2{n} \cdot \log^2{N})$ update time for the unweighted and $O(\exp(1/\eps) \log^2{n} \cdot \log^5{N} \cdot \log{W})$ update time for the weighted case. They cast interval scheduling as the problem of finding a maximum independent set among a set of intervals lying on the $x$-axis. The authors extend this setting to multiple dimensions and design algorithms for approximating maximum independent set among a set of $d$-dimensional hypercubes, achieving a $(1+\eps) 2^d$-approximation in the unweighted and a $(4 + \eps)2^d$-approximation in the weighted regime.

The authors of \cite{bhore2020dynamic} primarily focus on the unweighted case of approximating maximum independent set of a set of cubes. For the $1$-dimensional case, which equals interval scheduling on one machine, they obtain $O(\nicefrac{\log{n}}{\eps^2})$ update time, which is slower by a factor of $1/\eps$ than our approach. They also show that their approach generalizes to the $d$-dimensional case, requiring $\poly \log{n}$ amortized update time and providing $O(4^d)$ approximation.

\cite{gawrychowski2022sublinear} approach the problem of dynamically maintaining an exact solution to interval scheduling on one or multiple machines. They attain a guarantee of $\tilde{O}(n^{1/3})$ update time for unweighted interval scheduling on $M=1$ machine, and $\tilde{O}(n^{1-1/M})$ for $M \ge 2$. Moreover, they show an almost-linear time conditional hardness lower bound for dynamically maintaining an exact solution to the weighted interval scheduling problem on even just $M=1$ machine. This further motivates work such as ours that dynamically maintains approximate solutions for weighted interval scheduling.

\cite{gavruskin2014dynamic} consider dynamic interval scheduling on multiple machines in the setting in which all the jobs must be scheduled. The worst-case update time of their algorithm is $O(\log(n)+d)$, where $d$ refers to the depth of what they call \emph{idle intervals} (depth meaning the maximal number of intervals that contain a common point); they define an idle interval to be the period of time in a schedule between two consecutive jobs in a given machine. The same set of authors, in \cite{gavruskin2015dynamic_monotonic}, study dynamic algorithms for the monotone case as well, in which no interval completely contains another one. For this setup they obtain an algorithm with $O(\log(n))$ update and query time.

In the standard model of computing (i.e. one processor, static), there exists an $O(n+m)$ running time algorithm for (exactly) solving the unweighted interval scheduling problem on a single machine with $n$ jobs and integer coordinates bounded by $m$ \cite{frank1976some}.
An algorithm with running time independent of $m$ is described in \cite{tardos2005algorithm}, where it is shown how to solve this problem on $M$ machines in $O(n \log (n))$ time.
An algorithm is designed in \cite{arkin1987scheduling} for weighted interval scheduling on $M$ machines that runs in $O(n^2 \log(n))$ time.

We refer a reader to \cite{kolen2007interval} and references therein for additional applications of the interval scheduling problem.

\subparagraph{Other related work.}
There has also been a significant interest in job scheduling problems in which our goal is to schedule \emph{all} the given jobs across multiple machines, with the objective to minimize the total scheduling time. Several variants have been studied, including setups which allow preemptions, or setting where jobs have precedence constraints. We refer a reader to \cite{lenstra1978complexity,correa2005single,robert2008non,skutella2005stochastic,buttazzo2012limited,pinedo2012scheduling,levey20191} and references therein for more details on these and additional variants of job scheduling. Beyond dynamic algorithms for approximating maximum independent sets of intervals or hypercubes, \cite{cardinal2021worst} show results for geometric objects such as disks, fat polygons, and higher-dimensional analogs. After we had published a preprint of this work, \cite{cardinal2021worst} proved a result that captures \cref{theorem:unweighted-M=1} with a more general class of fat objects.

\section{Overview of Our Techniques}\label{section:techniques}
Our primary goal is to 
present unified techniques for approximating scheduling problems that can be turned into efficient algorithms for many settings. In this section, we discuss key insights of our techniques.

In the problems our work tackles, partitioning the problem instance into mostly-independent, manageable chunks is crucial. Doing so enables an LCA to determine information about a job of interest without computing an entire schedule, or enables a dynamic data structure to maintain a solution without restarting from scratch.

\subsection{Unweighted Interval Scheduling -- Partitioning Over Time (\cref{section:dynamic-unit})}
\label{sec:overview-unweighted}
For simplicity of presentation, we begin by examining our method for partitioning over time for just the unweighted interval scheduling problem on one machine (i.e., $M=1$). In particular, we first focus on doing so for the dynamic setting. 

Recall that in this setting the primary motivation for partitioning over time, is to divide the problem into independent, manageable chunks that can be utilized by a data structure to quickly modify a solution while processing an update. In our work, we partition the time dimension by maintaining a set of \emph{borders} that divide time into some number of contiguous regions.
By doing so, we divide the problem into many \emph{independent regions}, and we ignore jobs that intersect multiple regions; equivalently, we ignore jobs that contain a border.
Our goal is then to dynamically maintain borders in a way such that we can quickly recompute the optimal solution completely within some region, and that the suboptimality introduced by these borders does not affect our solution much.
In \cref{section:dynamic-unit}, we show that by maintaining borders where the optimal solution inside each region, i.e., a time-range between two borders, is of size $\Theta(\frac{1}{\eps})$, we can maintain a $(1+\eps)$-approximation of an optimal solution as long as we optimally compute the solution within each region.

Here, the underlying intuition is that because each region has a solution of size $\Omega(\frac{1}{\eps})$, we can charge any suboptimality caused by a border against the selected jobs in an adjacent region. Likewise, because each region's solution has size $O(\frac{1}{\eps})$, we are able to recompute the optimal solution within some region quickly using a balanced binary search tree. We dynamically maintain borders satisfying our desired properties by adding a new border when a region becomes too large, or merging with an adjacent region when a region becomes too small. As only $O(1)$ regions will require any modification when processing an update, this method of partitioning time, while simple, enables us to improve the fastest known update/query time to $O(\log(n)/\eps)$.
\footnote{The main advantage of this techniques is that it leads to worst-case $O(\log{(n)} / \eps)$ update time, as opposed to only an amortized one. We point out that it is not difficult to obtain such amortized guarantee in the following way: after each $\eps \cdot OPT$ many updates, recompute the optimum solution from scratch. Given access to the balanced binary tree structure described above, this re-computation can be done in $O(OPT \cdot \log n)$ time.}
In \cref{sec:overview-weighted} we build on these ideas to design an algorithm for the weighted interval scheduling problem.

\subsection{Weighted Interval Scheduling (\cref{section:dynamic-weighted})}
\label{sec:overview-weighted}
In our most technically involved result, we design the first deterministic $(1+\eps)$ approximation algorithm for weighted interval scheduling that runs in $\poly(\log n,\frac{1}{\eps})$ time. In this section we give an outline of our techniques and discuss key insights.
For full details we refer a reader to \cref{section:dynamic-weighted}.

\subsubsection{Job data structure (\cref{sec:hierarchical-decomposition})}
Let $\cE$ be the set of all the endpoints of given jobs, i.e., $\cE$ contains $s_i$ and $f_i$ for each job $[s_i, f_i]$. We build a hierarchical data structure over $\cE$ as follows.
This structure is organized as a binary search tree $T$.
Each node $Q$ of $T$ contains value $\key{Q} \in \cE$, with ``1-1'' mapping between $\cE$ and the nodes of $T$.
Each node $Q$ is \emph{responsible for a time range}. The root of $T$, that we denote by $Q_{root}$, is responsible for the entire time range $(-\infty, \infty)$.
    Each node $Q$ has at most two children, that we denote by $Q_L$ and $Q_R$. If $Q$ is responsible for the time range $[X, Y]$, then $Q_L$ is responsible for $[X, \key{Q}]$, while $Q_R$ is responsible for $[\key{Q}, Y]$.

Jobs are then assigned to nodes, where a job $J$ is assigned to every node $Q$ such that $J$ is contained within the $Q$'s responsible time range.

\begin{figure}[!ht] 
\centerline{\includegraphics[width=0.9\textwidth]{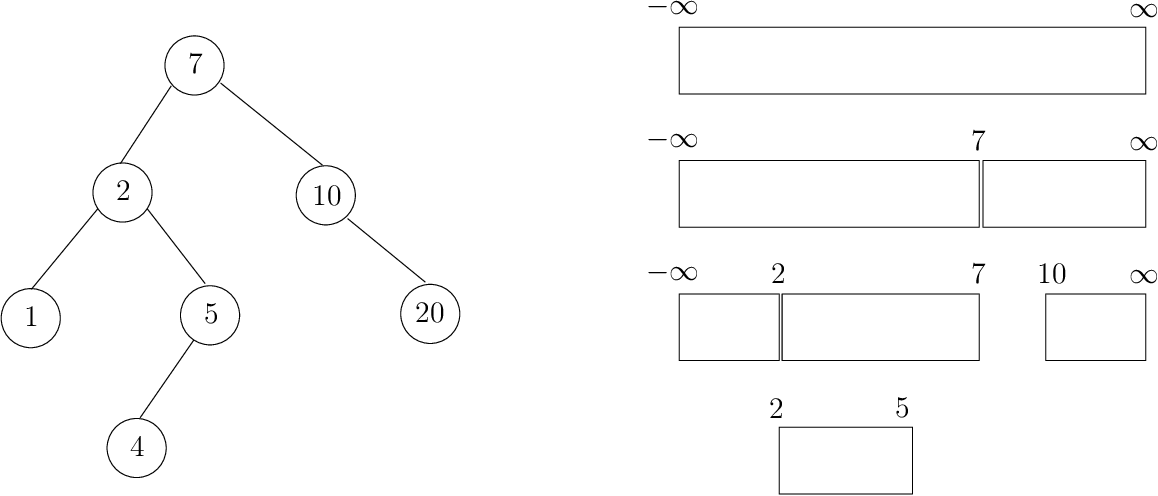}}
\caption{Visual example for hierarchical decomposition. Consider we are given jobs with the following ranges of $(1,5), (2,10), (7, 20), (4,5)$. On the left is $T$, a balanced binary search tree over the set of all $s_i$ and $f_i$. On the right is the hierarchical decomposition that corresponds to $T$.
That is, in each row, the intervals on the right correspond to the $[l_Q, r_Q]$ for the nodes on the left. For instance, in the third row, $(-\infty, 2]$ corresponds to the node $Q$ with $KEY(Q) = 1$.}
\label{fig:hierarchical}
\end{figure}

\subsubsection{Organizing computation (\cref{sec:convenient-structure})}
We now outline how the structure $T$ is used in computation. As a reminder, our main goal is to compute a $(1+\eps)$-approximate weighted interval scheduling. This task is performed by requesting $Q_{root}$ to solve the problem for the range $(-\infty, \infty)$. However, instead of computing the answer for the entire range $(-\infty, \infty)$ directly, $Q_{root}$ \emph{partitions} the range $(-\infty, \infty)$ into:
\begin{itemize}
    \item a number of ranges over which it is relatively easy to compute approximate solutions, such ones are called \emph{sparse}, and
    \item the remaining ranges over which it is relatively hard to compute approximate solutions at the level of $Q_{root}$.
\end{itemize}
These hard-to-approximate ranges are deferred to the children of $Q_{root}$, and are hard to approximate because any near-optimal solution for the range contains many jobs. On the other hand, solutions in sparse ranges are of size $O(1/\eps)$. As we discuss later, approximate optimal solutions within sparse ranges can be computed very efficiently; for details, see the paragraph \emph{Approximate dynamic programming} below.

In general, a child $Q_C$ of $Q_{root}$ might receive \emph{multiple} ranges from $Q_{root}$ for which it is asked to find an approximately optimal solution. $Q_C$ performs computation in the same manner as $Q_{root}$ did -- the cell $Q_C$ partitions each range it receives into ``easy'' and ``hard'' to compute subranges. The first type of subranges is computed by $Q_C$, while the second type if deferred to the children of $Q_C$. Here, ``hard'' ranges are akin to nodes having large solutions in our description of \cref{alg:alg-global-approx} in \cref{sec:time-partitioning-overview}. The same as in \cref{sec:time-partitioning-overview}, these ``hard'' ranges have large weight and allow for drawing a boundary and hence dividing a range into two or more \emph{independent} ranges. We now discuss how the partitioning into ranges is undertaken.

\subsubsection{Auxiliary data structure (\cref{sec:structure-Z(Q)})}
To divide a range into ``easy'' and ``hard'' ranges at the level of a node $Q$, we design an auxiliary data structure, which relates to a rough approximation of the problem. This structure, called $Z(Q)$, maintains a set of points (we call these points \emph{grid endpoints}) that partition $Q$ into \emph{slices of time}. We use slice to refer to a time range between two \emph{consecutive} points of $Z(Q)$. Recall how for unweighted interval scheduling, we maintained a set of borders and ignored a job that crossed any border. In the weighted version, we will instead use $Z(Q)$ as a set of partitions from which we will use \emph{some subset} to divide time. Our method of designing $Z(Q)$ reduces the task of finding a partitioning over time $Z(Q)$ within a cell for the $(1+\eps)$-approximate weighted interval scheduling problem to finding multiple partitionings for the  $(1+\eps)$-approximate unweighted problem.

It is instructive to think of $Z(Q)$ in the following way. First, we view weighted interval scheduling as $O(\log n)$ independent instances of unweighted interval scheduling -- instance $i$ contains the jobs having weights in the interval $(w_{max}(Q)/2^{i+1}, w_{max}(Q)/2^{i}]$. Then, for each unweighted instance we compute borders as described in \cref{sec:overview-unweighted}. $Z(Q)$ constitutes a subset of the union of those borders across all unweighted instances.
We point out that the actual definition of $Z(Q)$ contains some additional points that are needed for technical reasons, but in this section we will adopt this simplified view.
In particular, as we will see, $Z(Q)$ is designed such that the optimal solution within each slice has small total reward compared to the optimal solution over the entirety of $Q$. This enables us to partition the main problem into subproblems such that the suboptimality of discretizing the time towards slices, that we call \emph{snapping}, is negligible.

However, a priori, it is not even clear that such structure $Z(Q)$ exists. So, one of the primary goals in our analysis is to show that there exists a near-optimal solution of a desirable structure that can be captured by $Z(Q)$.
The main challenge here is to detect/localize sparse and dense ranges efficiently and in a way that yields a fast dynamic algorithm.
As an oversimplification, we define a solution as having \emph{nearly-optimal sparse structure} if it can be generated with roughly the following process:

\begin{itemize}
    \item Each cell $Q$ receives a set of disjoint time ranges for which it is supposed to compute an approximately optimal solution using jobs assigned to $Q$ or its descendants. Each received time range must have starting and ending time in $Z(Q)$.
    \item For each time range $\cR$ that $Q$ receives, the algorithm partitions $\cR$ into disjoint time ranges of three types: sparse time ranges, time ranges to be sent to $Q_L$ for processing, and time ranges to be sent to $Q_R$ for processing. In particular, this means that subranges of $\cR$ are deferred to the children of $Q$ for processing.
    \item For every sparse time range, $Q$ computes an optimal solution using at most $\nicefrac{1}{\eps}$ jobs.
    \item The union of the reward/solution of all sparse time ranges on all levels must be a $(1+\eps)$-approximation of the globally optimal solution without any structural requirements.
\end{itemize}

Moreover, and crucial for obtaining small running time per update, we develop a \emph{charging method} that enable us to partition each cell with only $|Z(Q)| = \poly(\nicefrac{1}{\eps},\log(n))$ points and still have the property that it contains a $(1+\eps)$-approximately optimal solution with nearly-optimal sparse structure. Then, we design an approximate dynamic programming approach to efficiently compute near-optimal solutions for sparse ranges. Combined, this enables a very efficient algorithm for weighted interval scheduling. On a high-level, $Z(Q)$ enables us to eventually decompose an entire solution into sparse regions.

\subsubsection{The charging method (\cref{sec:existence-of-nearly-optimal-sparse-structure})} We now outline insights of our charging arguments that enable us to convert an optimal solution $OPT$ into a near-optimal solution $OPT'$ with nearly-optimal sparse structure while relaxing our partitioning to only need $|Z(Q)| = \poly(\nicefrac{1}{\eps},\log(N))$ points. For a visual aid, see \cref{fig:charging}.

\begin{figure}[htbp]
\centerline{\includegraphics[width=0.9\textwidth]{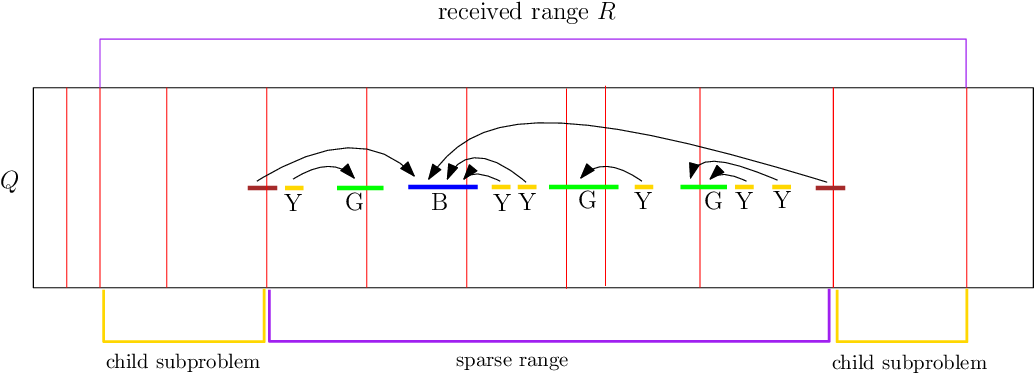}}
\caption{Visual example for charging argument.}
\label{fig:charging}
\end{figure}

As outlined in our overview of the nearly-optimal sparse structure, each cell $Q$ receives a set of disjoint time ranges, with each time range having endpoints in $Z(Q)$, and must split them into three sets: sparse time ranges, time ranges for $Q_L$, and time ranges for $Q_R$.
We will now modify $OPT$ by deleting some jobs. This new solution will be denoted by $OPT'$ and will have the following properties:
\begin{enumerate}[(1)]
    \item $OPT'$ exhibits nearly-optimal sparse structure; and
    \item $OPT'$ is obtained from $OPT$ by deleting jobs of total reward at most $O(\eps \cdot w(OPT))$.
\end{enumerate}
We outline an example of one such time range a cell $Q$ may receive in \cref{fig:charging}, annotated by ``received range $\cR$''. We will color jobs in \cref{fig:charging} to illustrate aspects of our charging argument, but note that jobs do not actually have a color property beyond this illustration.
Since our structure only allows a cell $Q$ to use a job within its corresponding time range, any relatively valuable job that crosses between $Q_L$ and $Q_R$ must be used now by $Q$ putting it in a sparse time range.
One such valuable job in \cref{fig:charging} is in blue marked by ``B''.
To have ``B'' belong to a sparse range, we must divide the time range $\cR$ somewhere, as otherwise our solution in the received range will be dense. If we naively divide $\mathcal{R}$ at the partition of $Z(Q)$ to the left and right of the job ``B'', we might be forced to delete some valuable jobs; such jobs are pictured in green and marked by ``G''.
Instead, we expand the division outwards in a more nuanced manner. Namely, we keep expanding outwards and looking at the job that contains the next partition point (if any). If the job's value exceeds a certain threshold, as those pictured as green and marked by ``G'' in \cref{fig:charging}, we continue expanding.
Otherwise, the job crossing a partition point is below a certain threshold, pictured as brown and not marked in \cref{fig:charging}, and its deletion can be charged against the blue job. We delete such brown jobs and the corresponding partition points, i.e., the vertical red lines crossing those brown jobs, constitute the start and the end of the sparse range.
By the end, we decided the starting and ending time of the sparse range, and what remains inside are blue job(s), green job(s), and yellow job(s) (also marked by ``Y''). Note that yellow jobs must be completely within a partition slice of $Z(Q)$.
Since we define $Z(Q)$ such that the optimal total reward within any grid slice is small, the yellow jobs have relatively small rewards compared to the total reward of green and blue jobs that we know must be large.
Accordingly, we can delete the yellow jobs (to help make this time range's solution sparse) and charge their cost against a nearby green or blue job.
In \cref{fig:charging}, an arrow from one job to another represents a deleted job pointing towards the job who we charge its loss against. Finally, each sparse range contains only green job(s) and blue job(s). If there are more than $\nicefrac{1}{\eps}$ jobs in such a sparse range, we  employ a simple sparsifying step detailed in the full proof. 

It remains to handle the time ranges of the received range that were not put in sparse ranges. These will be time ranges that are sent to $Q_L$ and $Q_R$. In \cref{fig:charging}, these ranges are outlined in yellow and annotated by ``child subproblem''. However, the time ranges do not necessarily align with $Z(Q_L)$ or $Z(Q_R)$ as is required by nearly-optimal sparse structure. We need to adjust these ranges such that they align with $Z(Q_L)$ or $Z(Q_R)$ so we can send the ranges to the children. See \cref{fig:snap} for intuition on why we cannot just immediately ``snap'' these child subproblems to the partition points in $Z(Q_L)$ and $Z(Q_R)$. (We say that a range $\mathcal{R}$ is \emph{snapped} inward (outward) within cell $Q$ if $\mathcal{R}$ is shrunk (extended) on both sides to the closest points in $Z(Q)$. Inward snapping is illustrated in \cref{fig:snap}.) Instead, we employ a similar charging argument to deal with snapping. As an analog to how we expanded outwards from the blue job for defining sparse ranges, we employ a charging argument where we contract inwards from the endpoints of the child subproblem. In summary, these charging arguments enabled us to show a solution of nearly-optimal sparse structure exists even when only partitioning each cell $Q$ with $|Z(Q)| = \poly(\nicefrac{1}{\eps},\log(n))$ points.

\subsubsection{Approximate dynamic programming (\cref{sec:dp-sparse})} Now, we outline our key advance for more efficiently calculating the solution of nearly-optimal sparse structure. This structure allows us to partition time into ranges with sparse solutions. More formally, we are given a time range and we want to approximate an optimal solution within that range that uses at most $\nicefrac{1}{\eps}$ jobs.
We outline an approximate dynamic programming approach that only requires polynomial time dependence on $\nicefrac{1}{\eps}$.

The relatively well-known dynamic programming approach for computing weighted interval scheduling is to maintain a dynamic program where the state is a prefix range of time and the output is the maximum total reward that can be obtained in that prefix range of time. However, for our purposes, there are too many possibilities for prefix ranges of time to consider. Instead, we invert the dynamic programming approach, and have a state referencing some amount of reward, where the dynamic program returns the minimum length prefix range of time in which one can obtain a given reward. Unfortunately, there are also too many possible amounts of rewards. We observe that we do not actually need this exact state, but only an approximation. In particular, we show that one can round this state down to powers of $(1+\eps^2)$ and hence significantly reduce the state-space. In \cref{sec:dp-sparse}, we show how one can use this type of observation to quickly compute approximate dynamic programming for a near-optimal sparse solution inside any time range.

\subsubsection{Comparison with Prior Work} \label{subsection:compare}
The closest to our work is the one of \cite{henzinger2020dynamic}. In terms of improvements, we achieve the following: we remove the dependence on $N$ and $w_{\rm max}$ in the running-time analysis; we obtain a deterministic approach; and, we design an algorithm with $\poly (1/\eps, \log n)$ update/query time, which is exponentially faster in $1/\eps$ compared the prior work.

In this prior work, jobs are assumed to have length at least $1$ and belong in the time-interval $[1,N]$. To remove the dependence on $N$ and such assumptions, we designed a new way of bookkeeping jobs. Instead of using a complete binary tree on $[1, N]$ to organize jobs as done in the prior work, we employ binary balanced search tree on the endpoints of jobs. A complete binary tree on $[1, N]$ is oblivious to the density of jobs. On the other hand, and intuitively, our approach allows for ``instance-based'' bookkeeping: the jobs are in a natural way organized with respect to their density. Resorting to this approach incurs significant technical challenges. Namely, the structure of solution our tree maintains is hierarchically organized. However, each tree update, which requires node-rotations, breaks this structure which requires additional care in efficiently maintaining approximate solution after an update, as well as requiring an entirely different approach for maintaining a partitioning of time $Z(Q)$ within cells. 
Moreover, we show how to further leverage these ideas to obtain a deterministic approach.

In our work, we use borders to define the so-called sparse and dense ranges. This idea is inspired by the work of \cite{henzinger2020dynamic}. We emphasize, though, that one of our main contributions and arguably the most technically involved component is showing how to algorithmically employ those borders in running-time only polynomially dependent on $1/\eps$, while \cite{henzinger2020dynamic} require exponential dependence on $1/\eps$.

Our construction of auxiliary data structure $Z(Q)$ enables us to boost an $O(\log(n))$-approximate solution into a decomposition enabling a $(1+\eps)$-approximate solution is inspired by the approach of \cite{henzinger2020dynamic}. They similarly develop $Z(Q)$ to boost an instead $O(1)$-approximation that fundamentally relies on the bounded coordinate assumptions of jobs being within $[1,N]$ and having length at least 1. Our different approach towards $Z(Q)$ enables simplification of some arguments as well as not relying on randomness, or on length or bounded coordinate assumptions. Further, we note that the dynamic programming approach for sparse regions that we develop is significantly faster than the enumerative approach used in the prior work, that eventually enables us to obtain a $\poly (1/\eps)$ dependence in the running time. The way we combine solutions over sparse regions is similar to the way it is done in the prior work.

\subsection{Localizing the Time-Partitioning Method (\cref{section:local})}
\label{sec:time-partitioning-overview}
We also show that this method of partitioning over time can be used to develop local algorithms for interval scheduling. Here, we desire to answer queries about whether a particular job is in our schedule. We hope to answer each of these queries consistently (i.e., they all agree with some approximately optimal schedule) and in less time than it would take to compute an entire schedule from scratch. Partitioning over time seems helpful for this setting, because this would enable us to focus on just the region of the job being queried. However, our previously mentioned method for maintaining borders does so in a sequential manner that we can no longer afford to do in this model of computation. Instead, we use a hierarchical approach to more easily compute the locations of borders that create regions with solutions not too big or too small.

For simplicity, we again focus on the unweighted setting with only one machine. In the standard greedy algorithm for computing unweighted interval scheduling on one machine, we repeatedly select the job $successor(x)$: ``\emph{What is the interval with the earliest endpoint, of those that start after point $x$?}'' (where $x$ is the endpoint of the previously chosen job). As reading the entire problem instance would take longer than desired, an LCA requires some method of probing for information about the instance. Our LCA utilizes such successor probes to do so.
For further motivation, see \cref{section:local}. We outline a three-step approach towards designing an LCA that utilizes few probes:

\emph{Hierarchizing the greedy (\cref{alg:alg-global-exact}).} Instead of just repeatedly using $successor(x)$ to compute the solution as the standard greedy does, we add hierarchical structure that adds no immediate value but serves as a helpful stepping stone.
Consider a \emph{binary search tree} (BST) like structure, where the root node corresponds to the entire time range $[0, N]$. Each node in the structure has a left-child and a right-child corresponding to the 1st and the 2nd half, respectively, of that node's range. Eventually, leaf nodes have no children and correspond to a time range of length one unit. At a high-level, we add hierarchical structure by considering jobs contained in some node's left-child, then considering jobs that go between the node's left-child and right-child, and then considering jobs contained in the node's right-child. This produces the same result as the standard greedy, but we do so with a hierarchical structure that will be easier to utilize.

\emph{Approximating the hierarchical greedy (\cref{alg:alg-global-approx}).}
Now, we modify the hierarchical greedy so that it is no longer exactly optimal but is instead an approximation. At first this will seem strictly worse, but it will yield an algorithm that is easier to localize. When processing each node, we will first check whether it is the case that both the left-child and the right-child have optimal solutions of size $>\frac{1}{\eps}$. A key observation here is that checking whether a time range has an optimal solution of size $>\frac{1}{\eps}$ can be done by making at most $1+\frac{1}{\eps}$ successor probes (i.e., one does not necessarily need to compute the entire optimal solution to check if it is larger than some relatively small threshold).
If both the left-child and the right-child would have optimal solutions of size $>\frac{1}{\eps}$, then we can afford to draw a border at the midpoint of our current node and solve the left-child and right-child independently. Jobs intersecting a border are \emph{ignored}, and we charge the number of such ignored jobs, i.e., the number of drawn borders, to the size of solution in the corresponding left- and right-child. Ultimately, we show that the addition of these borders makes our algorithm $(1+\eps)$-approximate. Moreover, and importantly, these borders introduce \emph{independence} between children with large solutions.

\emph{Localizing the approximate, hierarchical greedy (\cref{alg:alg-local-approx}).} Finally, we localize the approximate, hierarchical greedy. To do so, we note that when some child of a node has a small optimal solution, then we can get all the information we need from that child in $O(\frac{1}{\eps})$ probes. As such, if a node has a child with a small optimal solution, we can make the required probes from the small child and recurse to the large child. Otherwise, if both children have large solutions, we can draw a border at the midpoint of the current node and only need to recurse down the child which contains the job the LCA is being queried about. 

With these insights, we have used our partitioning method over time for local algorithms to produce an LCA only requiring $O(\frac{\log(N)}{\eps})$ successor probes.

\section{Problem Setup}\label{section:setup}
In the interval scheduling problem, we are given $n$ jobs and $M$ machines. With each job $j$ are associated two numbers $s_j$ and $l_j > 0$, referring to ``start'' and ``length'' respectively, meaning that the job $j$ takes $l_j$ time to be processed and its processing can only start at time $s_j$. While prior work such as \cite{henzinger2020dynamic} used assumptions such as $s_j \ge 0, l_j \ge 1$ and have an upper-bound $N$ on $s_j + l_j$, we utilize such assumptions \emph{only in our LCA results of \cref{section:local}}. In addition, with each job $j$ is associated weight/reward $w_j > 0$, that refers to the reward for processing the job $j$. The task of \emph{interval scheduling} is to schedule jobs across machines while maximizing the total reward and respecting that each of the $M$ machines can process at most one job at any point in time.

\section{Dynamic Unweighted Interval Scheduling on a Single Machine}\label{section:dynamic-unit}

In this section we prove \cref{theorem:unweighted-M=1}.
As a reminder, \cref{theorem:unweighted-M=1} considers  the case of interval scheduling in which $w_j = 1$ for each $j$ and $M = 1$, i.e., the jobs have unit reward and there is only a single machine at our disposal. This case can also be seen as a task of finding a maximum independent set among intervals lying on the $x$-axis. The crux of our approach is in designing an algorithm that maintains the following invariant:

\smallskip
\begin{minipage}{0.95\linewidth}
\begin{mdframed}[hidealllines=true, backgroundcolor=gray!15]
\vspace{-3pt} 
\begin{invariant}\label{invariant:unweighted-MIS}
    The algorithm maintains a set of borders such that an optimal solution schedules between $\nicefrac{1}{\varepsilon}$ and $\nicefrac{2}{\varepsilon}$ intervals within each two consecutive borders.
\end{invariant}
\end{mdframed}
\end{minipage}
\medskip 

We will maintain this invariant unless the optimal solution has fewer than $\nicefrac{1}{\eps}$ intervals, in which case we are able to compute the solution from scratch in negligible time. We aim for our algorithm to maintain \cref{invariant:unweighted-MIS} while keeping track of the optimal solution between each pair of consecutive borders. The high level intuition for this is that if we do not maintain too many borders, then our solution must be very good (our solution decreases by size at most one every time we add a new border). Furthermore, if the optimal solution within borders is small, it is likely easier for us to maintain said solutions. We prove that this invariant enables a high-quality approximation:
\begin{lemma}\label{lemma:invariant-value}
 A solution that maintains an optimal solution within consecutive pairs of a set of borders, where the optimal solution within each pair of consecutive borders contains at least $K$ intervals, maintains a $\frac{K+1}{K}$-approximation.
\end{lemma}
\begin{proof}
For our analysis, suppose there are implicit borders at $-\infty$ and $+\infty$ so that all jobs are within the range of borders. Consider an optimal solution $OPT$. We will now design a $K$-approximate optimal solution $OPT'$ as follows: given $OPT$, delete all intervals in $OPT$ that overlap a drawn border.
Fix an interval $J$ appearing in $OPT$ but not in $OPT'$. Assume that $J$ intersects the $i$-th border. Recall that between the $(i-1)$-st and the $i$-th border there are at least $K$ intervals in $OPT'$. Moreover, at most one interval from $OPT$ intersects the $i$-th border. Hence, to show that $OPT'$ is a $\frac{K+1}{K}$-approximation of $OPT$, we can charge the removal of $J$ to the intervals appearing between the $(i-1)$-st and the $i$-th border in $OPT'$.
\end{proof}

Not only does \cref{invariant:unweighted-MIS} enable high-quality solutions, but it also assists us in quickly maintaining such a solution. We can maintain a data structure with $O(\frac{\log(n)}{\varepsilon})$ updates and $O(\log(n))$ queries that moves the borders to maintain the invariant and thus maintains an $(1+\varepsilon)$-approximation as implied by \cref{lemma:invariant-value}.

\theoremunweightedM*

\begin{proof}
Our goal now is to design an algorithm that maintains \cref{invariant:unweighted-MIS}, which by \cref{lemma:invariant-value} and for $K = \nicefrac{1}{\eps}$ will result in a $(1+\varepsilon)$-approximation of \MaxIS.

On a high-level, our algorithm will maintain a set of borders. When compiling a solution of intervals, the algorithm will not use any interval that contains any of the borders, but proceed by computing an optimal solution between each two consecutive borders. The union of those between-border solutions is the final solution.
Moreover, we will maintain the invariant that the optimal solution for every contiguous region is of size within $[\frac{1}{\varepsilon}, \frac{2}{\varepsilon})$.

In the rest, we show how to implement these steps in the claimed running time.

\subparagraph{Maintained data-structures.}

 Our algorithm maintains a balanced binary search tree $T_{\rm{all}}$ of intervals sorted by their starting points. Each node of $T_{\rm{all}}$ will also maintain the end-point of the corresponding interval. It is well-known how to implement a balanced binary search tree with $O(\log n)$ worst-case running time per insertion, deletion and search query. Using such an implementation, the algorithm can in $O(\log n)$ time find the smallest ending-point in a prefix/suffix on the intervals sorted by their starting-points. That is, in $O(\log{n})$ time we can find the interval that ends earliest, among those that start after a certain time.

In addition, the algorithm also maintains a balanced binary search tree $T_{\rm{borders}}$ of the borders currently drawn.

Also, we will maintain one more balanced binary search tree $T_{\rm{sol}}$ that will store the intervals that are in our current solution.

We will use that for any range with optimal solution of size $S$, we can make $O(S)$ queries to these data structures to obtain an optimal solution for the range in $O(S \cdot \log n)$ time.

\subparagraph{Update after an insertion.}
Upon insertion of an interval $J$, we add $J$ to $T_{\rm{all}}$. We make a query to $T_{\rm{borders}}$ to check whether $J$ overlaps a border. If it does, we need to do nothing; in this case, we ignore $J$ even if it belongs to an optimal solution. If it does not, we recompute the optimal solution within the two borders adjacent to $J$. If after recomputing, the new solution between the two borders is too large, i.e, it has at least $\frac{2}{\varepsilon}$ intervals, then draw/add a border between the $\frac{1}{\varepsilon}$-th and the $(1+\frac{1}{\varepsilon})$-th of those intervals.

\subparagraph{Update after a deletion.}
Upon deletion of an interval $J$, we delete $J$ from $T_{\rm{all}}$. If $J$ was not in our solution, we do nothing else. Otherwise, we recompute the optimal solution within the borders adjacent to $J$ and modify $T_{\rm{sol}}$ accordingly. Let those borders be the $i$-th and the $(i+1)$-st. If the new solution between borders $i$ and $i+1$ now has size less than $\nicefrac{1}{\varepsilon}$ (it would be size exactly $\nicefrac{1}{\eps}$), we delete an arbitrary one of the two borders (thus combining this region with an adjacent region). Then, we recompute the optimal solution within the (now larger) region $J$ is in.  If this results in a solution of size at least $\nicefrac{2}{\varepsilon}$, we will need to split the newly created region by adding a border. Before splitting, the solution will have size upper-bounded by one more than the size of the solutions within the two regions before combining them as an interval may have overlapped the now deleted border (one region with size exactly $\frac{1}{\eps}-1$ and the other upper-bounded by $\frac{2}{\eps}-1$). Thus, the solution has size at in range $[\nicefrac{2}{\eps},\frac{3}{\eps})$. We can add a border between interval $\nicefrac{1}{\eps}$ and $\nicefrac{1}{\eps}+1$ of the optimal solution, and will have a region with exactly $\nicefrac{1}{\eps}$ intervals and another with $[\nicefrac{1}{\eps},\nicefrac{2}{\eps})$ intervals, maintaining our invariant.

In all of these, the optimal solution for each region has size $O(\nicefrac{1}{\varepsilon})$, so recomputing takes $O(\nicefrac{\log (n)}{\varepsilon})$ time.

For queries, we will have maintained $T_{\rm{sol}}$ in our updates such that it contains exactly the intervals in our solution. So each query we just need to do a lookup to see if the interval is in $T_{\rm{sol}}$ in $O(\log n)$ time.
\end{proof}

This result improves the best-known time complexities \cite{bhore2020dynamic,henzinger2020dynamic}. Unfortunately, it does not immediately generalize well to the weighted variant. In \cref{section:dynamic-weighted}, we show our more technically-challenging result for the weighted variant.

\section{Dynamic Weighted Interval Scheduling on a Single Machine}\label{section:dynamic-weighted}
This section focuses on a more challenging setting in which jobs have non-uniform weights. Non-uniform weights introduce difficulties for the approach mentioned in \cref{section:dynamic-unit}, as adding a border (which entails ignoring all the jobs that cross that border) may now force us to ignore a very valuable job. Straightforward extensions of this border-based approach require at least a linear dependence on the ratio between job rewards (e.g., if all jobs have rewards within $[1,w]$, then straightforward extensions would require a linear dependence on $w$). This is because an ignored job containing a border can have a reward of $w$ (as opposed to just $1$), requiring $\nicefrac{w}{\eps}$ reward inside the region to charge it against (as opposed to just $\nicefrac{1}{\eps}$). In this work, we show how to perform this task in $O(\textrm{poly}(\log(n),\nicefrac{1}{\varepsilon}))$ time, having no such dependency on the rewards of the jobs or the starting/ending times.
This improves upon the best-known preexisting result of $O(\textrm{poly}(\log(n),\log(N),\log(w)) \cdot \textrm{exp}(\nicefrac{1}{\varepsilon}))$ time accomplished by the decomposition scheme designed in the work of Henzinger et al.~\cite{henzinger2020dynamic}, which we compare with in \cref{subsection:compare}.
Both our algorithm and our analysis introduce new ideas that enable us to design a dynamic algorithm with running time having only polynomial dependence on $\nicefrac{1}{\eps}$ and $\log(n)$, yielding an exponential improvement in terms of $\nicefrac{1}{\eps}$ over \cite{henzinger2020dynamic}, and removing all dependence on $N$ and $w$. Moreover, our algorithm is deterministic (as opposed to randomized and a $(1+\varepsilon)$-approximation in expectation) and requires no assumption on the lengths or coordinate values of the jobs (\cite{henzinger2020dynamic} assumes all jobs are length at least $1$ and all coordinates are within $[0,N]$, where $N$ affects the time complexity).

As the first step we show that there exists a solution $OPT'$, which is a $(1+\eps)$-approximate optimal solution, that has \emph{nearly-optimal sparse structure}, similar to a structure used in \cite{henzinger2020dynamic}. We define properties of this structure in \cref{sec:convenient-structure}, although it is instructive to think of this structure as of a set of non-overlapping time ranges such that:
\begin{enumerate}[(1)]
    \item Within each time range there is an approximately optimal solution which contains a small number of jobs (called \emph{sparse}); 
    \item The union of solutions across all the time ranges is $(1+\eps)$-approximate; and
    \item There is an efficient algorithm to obtain these time ranges.
\end{enumerate}
Effectively, this structure partitions time such that we get an approximately optimal solution by computing sparse solutions within partitioned time ranges and ignoring jobs that are not fully contained within one partitioned time range. To obtain the guarantees of such a set of time ranges that can be obtained efficiently, we utilize a new hierarchical decomposition based on a balanced binary search tree and employ novel charging arguments. This result is described in detail in \cref{sec:convenient-structure}.

Once equipped with this  structural result, we first design a dynamic programming approach to compute an approximately optimal solution within one time range. Let $w_{max}$ denote the maximal reward among all jobs currently in the instance. To obtain an algorithm whose running time is proportional to the number of jobs in the solution for a time range,
as opposed to the length of that range, we ``approximate'' states that our dynamic programming approach maintains, and ultimately obtain the following claim whose proof is deferred to \cref{sec:dp-sparse}.
\begin{restatable}{lemma}{lemmadpsparse}
\label{lemma:sparse-sol-approx}
 Given any contiguous time range $\mathcal{R}$ and an integer $K$, consider an optimal solution $OPT(\mathcal{R}, K)$ in $\mathcal{R}$ containing at most $K$ jobs and ignoring jobs with weight less than $\nicefrac{\eps}{n} \cdot w_{max}$. Then, there is an algorithm that in $\mathcal{R}$ finds a $(1+\eps)$-approximate solution to $OPT(\mathcal{R}, K)$ in  
 $O\rb{\frac{K \log(n) \log^2(K/\eps)}{\eps^2}}$ 
 time and with at most $O\rb{\frac{K \log(K/\varepsilon)}{\varepsilon}}$ jobs.
\end{restatable}
Observe that running time of the algorithm given by \cref{lemma:sparse-sol-approx} has no dependence on the length of $\mathcal{R}$. Also observe that the algorithm possibly selects slightly more than $K$ jobs to obtain a $(1+\eps)$-approximation of the best possible reward one could obtain by using at most $K$ jobs in $\mathcal{R}$ (i.e., $OPT(\mathcal{R}, K)$).

Finally, in \cref{sec:combining-ingredients} we combine all these ingredients and prove the main theorem of this section.
\theoremweighteddynamic*

\subsection{Decomposition Overview}
\label{sec:decomposition-overview}

We utilize a hierarchical decomposition to organize time such that we may efficiently obtain time ranges that satisfy the nearly-optimal sparse structure. This decomposition has two levels of granularity. For the higher-level decomposition, we employ a decomposition similar to that of a balanced binary search tree with $O(\log(n))$ depth. Each cell $Q$ in this balanced binary search tree will correspond to a range of time. Further details on this hierarchical decomposition are described in \cref{sec:hierarchical-decomposition}. 

For the lower-level decomposition, we split each cell $Q$ more finely. Formally, for a set of grid endpoints $Z(Q)$, we define a grid slice as follows.
\begin{definition}[Grid slice]\label{definition:grid-slice}
    Given a set of grid endpoints $Z(Q) = \{r_1, r_2, \ldots, r_{X-1}\}$ with $r_i < r_{i + 1}$, we use \emph{grid slice} to refer to an interval $(r_i, r_{i + 1})$, for any $1 \le i < {X-1}$. Note that a grid slice between $r_i$ and $r_{i + 1}$ does not contain $r_i$ nor $r_{i + 1}$.
\end{definition}

We further discuss $Z(Q)$ in \cref{sec:structure-Z(Q)}. Importantly, $Z(Q)$ is designed such that the optimal solution entirely within any grid slice is upper-bounded to be relatively small compared to the weight of the optimal solution within $Q$, or $w(OPT(Q))$.
This property makes the grid endpoints $Z(Q)$ a helpful tool in partitioning time. At a high level, $Z(Q)$ is used to define a set of segments that motivate dynamic programming states of the form $DP(Q,S)$, where each $S$ corresponds to a segment between two grid endpoints of $Z(Q)$, and $DP(Q,S)$ computes an approximately optimal sparse solution among schedules that can only use jobs contained within the segment of time $S$. The key idea is that this dynamic programming enables the partitioning of time into \emph{dense} and \emph{sparse} ranges. Solutions for sparse ranges are computed immediately, while dense ranges are solved by children with dynamic programming (by further dividing the dense range into more sparse and dense ranges). We recall from \cref{subsection:compare} that \cite{henzinger2020dynamic} were first to design a two-level hierarchical decomposition that computes $DP(Q,S)$ to optimize over dense and sparse ranges. However, we emphasize that our work utilizes entirely new approaches for our high-level hierarchical decomposition into cells $Q$, for our low-level decomposition of each cell into $Z(Q)$, and for our method of computing approximately optimal sparse solutions of $DP(Q,S)$.

\subsection{Solution of Nearly-Optimal Sparse Structure}
\label{sec:convenient-structure}
To remove exponential dependence on $\nicefrac{1}{\eps}$ and all dependence on $N$ and $w$, we introduce a new algorithm for approximating sparse solutions, a new hierarchical decomposition, and novel charging arguments that (among other things) reduce the number of grid endpoints $|Z(Q)|$ required in each cell. With this, we will compute an approximately optimal solution of the following very specific structure.

\begin{definition}[Nearly-optimal sparse structure]
    To have nearly-optimal sparse structure, a solution must be able to be generated with the following specific procedure:

\begin{itemize}
    \item Each cell $Q$ will receive a set of time ranges, denoted as $RANGES(Q)$, with endpoints in $Z(Q)$. To start, $Q_{root}$ will receive one time range containing all of time (i.e., $RANGES(Q_{root}) = \{[-\infty,\infty]\}$)
    \item $RANGES(Q)$ is split into a collection of disjoint time ranges, with each being assigned to one of three sets: $SPARSE(Q)$, $RANGES(Q_L)$, $RANGES(Q_R)$
    \item $SPARSE(Q)$, a set of time ranges, must have endpoints in $Z(Q) \cup Z(Q_L) \cup Z(Q_R)$
    \item For each child $Q_{child}$ (where $child \in \{L,R\}$) of $Q$, $RANGES(Q_{child})$ must have all endpoints in $Z(Q_{child})$
    \item The total weight of sparse solutions (solutions with at most $\nicefrac{1}{\eps}$ jobs) within sparse time ranges must be large (where $SPARSE\_OPT(\mathcal{R})$ denotes an optimal solution having at most $\nicefrac{1}{\eps}$ jobs within range $\mathcal{R}$): \[ \sum_Q \sum_{R \in SPARSE(Q)} w(SPARSE\_OPT(\mathcal{R})) \ge (1-O(\eps))w(OPT)
    \]
\end{itemize}

\end{definition}

Now, we prove our result for a $(1+\varepsilon)$-approximation to dynamic, weighted interval \MaxIS algorithm with only polynomial time dependence on $\nicefrac{1}{\eps}$ and $\log(n)$.
Unlike the decomposition of Henzinger et al., we will not define our decomposition such that each cell $Q$ will split exactly in half to produce both its children $Q_L$ and $Q_R$. Instead, we will divide every cell $Q$ in a manner informed by a balanced binary search tree. Desirably, this will make the depth of our decomposition $O(\log(n))$ instead of $O(\log(N))$, but it will remove the possibility of utilizing the random-offset style of idea to assign jobs to cells where they each job's length is approximately $\varepsilon$ fraction of the cell's length. This necessitates novel charging arguments. We supplement this new hierarchical decomposition with a new alternative for the $Z(Q)$ data structure that enables us to determine important dynamic program subproblems without any dependence on $N$. Additionally, we take a new approach for solving the small sparse subproblems, where we use an approximate dynamic programming idea to remove exponential dependence on $\nicefrac{1}{\eps}$ in the best known running time for these subproblems.
In our novel charging arguments, there is a particular focus on changing where deleted intervals' weights are charged against and introducing a \emph{snapping budget}, which we use to relax the required number of grid endpoints $|Z(Q)|$ to depend only polynomially on $\nicefrac{1}{\eps}$. As a reminder, $Z(Q)$ is a set of grid points within $Q$ such that between any two consecutive points we are guaranteed that the optimal solution has small weight. 
Our final algorithm will consider a number of subproblems for each cell proportional to $|Z(Q)|^2$, so improvements in $|Z(Q)|$ directly lead to improvements in the best-known running time.
Effectively, we make each of our smaller subproblems easier to solve while also reducing the number of subproblems we need to solve.  All improvements are exponential in $\eps$ and remove dependence on $N$ and $w$.

\subsubsection{Hierarchical decomposition}
\label{sec:hierarchical-decomposition}
We now formally describe our hierarchical decomposition of jobs.

\begin{itemize}
    \item Consider the set of all jobs' starting/ending times, i.e., for each job $i$, include $s_i$ and $f_i$. Now, consider a balanced binary search tree $T$ over this set of times. For the sake of this paper, one can assume this is maintained by a red-black tree such that the tree has depth $O(\log(n))$ and $O(\log(n))$ rotations are required per update. We have a cell $Q$ in our hierarchical decomposition corresponding to each node in $T$. Let $KEY(Q)$ be the corresponding key for the node in $T$.
    
    \item Each $Q$ has a left child $Q_L$ or right child $Q_R$ if the corresponding node in $T$ does.
    \item Each cell $Q$ represents a range of time. $Q_{root}$ corresponds to all time, meaning $TIME(Q_{root})=[-\infty,\infty]$. This time range is split for the children of $Q$ by $KEY(Q)$. More formally, given a cell $Q$ where $TIME(Q)=[l_Q,r_Q]$, then (if $Q_L$ exists) $TIME(Q_L)=[l_Q,KEY(Q)]$, and (if $Q_R$ exists) $TIME(Q_R)=[KEY(Q),r_Q]$.
\end{itemize}

This fully describes our hierarchical decomposition of depth $O(\log(n))$. A visual example is provided in \cref{fig:hierarchical}.

\subsubsection{Structure $Z(Q)$}
\label{sec:structure-Z(Q)}
We use the set of grid points $Z(Q)$ to determine segments that will be used as subproblems for dynamic programming and in reference to the nearly-optimal sparse structure. For some specified $X$, our goal is to maintain a $Z(Q)$ such that the optimal solution within every grid slice is at most $O(\nicefrac{w(OPT(Q))}{X})$. The previously-utilized methods for obtaining this require logarithmic dependence on $N$ and $w$. To remove dependence on $w$, we relax our requirements of $Z(Q)$ to ignore all jobs with weight less than $w(OPT(Q)) \cdot \nicefrac{\eps}{n}$; in total, these jobs have negligible reward. To remove dependence on $N$, we consider an alternative approach to computing $Z(Q)$, where we take the union of multiple solutions to $Z(Q)$ for the analogous unweighted interval scheduling problem using ideas similar to those in \cref{section:dynamic-unit}. We design a $Z(Q)$ with the following guarantees, whose proof is deferred to \cref{subsection:maintaining-z}:

\begin{restatable}[Dynamically maintaining $Z(Q)$]{lemma}{lemmazqmaintain}
\label{lemma:zq-maintain}
For any fixed positive integer $X$, it is possible to return a set $Z(Q)$ for any cell $Q$ in the hierarchical decomposition in $O(X \cdot \log^3(n))$ query time. Moreover, the returned $Z(Q)$ will satisfy the following properties:
\begin{itemize}
    \item For every $Q$, the optimal solution within each grid slice of $Z(Q)$ is at most $O(\nicefrac{w(OPT(Q))}{X})$; as a reminder, we ignore jobs with weights less than $w(OPT(Q)) \cdot \nicefrac{\eps}{n}$. 
    \item For every $Q$, $|Z(Q)| = O(X \cdot \log^2(n))$
\end{itemize}
\end{restatable}

\subsubsection{Existence of desired $(1+\eps)$-approximate solution}
\label{sec:existence-of-nearly-optimal-sparse-structure}
We now argue that there exists a $(1+O(\eps))$-approximation with nearly-optimal sparse structure in reference to our new hierarchical decomposition for $Q$ and our $Z(Q)$ when using $X=\frac{\log^2(n)}{\eps^2}$ and thus $|Z(Q)|=O(\frac{\log^4(n)}{\eps^2})$:

\begin{lemma} \label{lemma:sol-structure}
There exists a solution $OPT'$ that has nearly-optimal sparse structure and such that $w(OPT') \ge (1 - O(\varepsilon)) w(OPT)$. Thus, $OPT'$ is a $(1 + O(\varepsilon))$-approximation of $OPT$.
\end{lemma}

\begin{proof}
We emphasize that the goal of this lemma is not to show how to construct a solution algorithmically, but rather to show that there exists one, that we refer to by $OPT'$, that has a specific structure and whose weight is close to $OPT$.

In this paragraph, we provide a proof overview. At a high-level, we show this claim by starting with $OPT$, and maintaining a solution $OPT'$ that holds our desired structure and only deletes jobs with total weight $O(\varepsilon \cdot w(OPT))$. Our process of converting $OPT$ to $OPT'$ is recursive, as we start at the root and work down. Generally, our preference for any range $\mathcal{R} \in RANGES(Q)$ will be to defer it to a child by passing it on to a $RANGE(Q_{child})$. This preference can often not be immediately satisfied for two reasons: (i) $\mathcal{R}$ may not be completely contained within a $Q_{child}$ (i.e. $\mathcal{R}$ crosses between $Q_L$ and $Q_R$), or (ii) the endpoints of $\mathcal{R}$ do not alight with the corresponding $Z(Q_{child})$. We will modify $OPT$ to accommodate these concerns.
To handle concern (i), we will delete a job in $OPT$ if it crosses between $Q_{L}$ and $Q_{R}$ and has small went (and hence it can be ignored). Otherwise, if such a crossing job has large weight, we will divide $\mathcal{R}$ into three time ranges such that one is contained within $Q_L$, one uses the crossing job, and the last is contained within $Q_R$, using a process detailed in the following proof. For the central third, we will sparsify this range to produce a set $SPARSE(Q)$ of sparse time ranges.
For time ranges completely contained within $Q_L$ and $Q_R$ that are not designated as sparse time ranges, we will essentially consider them dense time ranges, that will be delegated to children cells of $Q$. In order to delegate a time range to a child $Q_{child}$, we require that the delegated time range must have endpoints that align with $Z(Q_{child})$. Accordingly, we perform modifications to ``snap'' the time ranges' endpoints to $Z(Q_{child})$ for the corresponding child $Q_{child}$ of $Q$ and include the ``snapped'' time ranges in $RANGES(Q_{child})$. We show that throughout this process, we do not delete much weight from $OPT$ and obtain an $OPT'$ that has our desired structure. Now, we present the proof in detail:

\paragraph{Deleting light crossing jobs.}
We now describe how to modify $OPT$, obtaining $OPT'$, such that $OPT'$ has our desired structure and $OPT'$ is a $(1+\eps)$-approximation of $OPT$. Note that we will never actually compute $OPT'$. It is only a hypothetical solution that has nice structural properties and that we use to compare our output to.

For a cell $Q$, consider a time range it receives in $RANGES(Q)$. We shall split this time range into sparse time ranges (to be added to $SPARSE(Q)$) and dense time ranges (to be added to $RANGES(Q_L)$ or $RANGES(Q_R)$). There is at most one range $\mathcal{R}_{cross} \in RANGES(Q)$ that crosses between $Q_L$ and $Q_R$, and we call the at most one job crossing between $Q_L$ and $Q_R$ the \emph{crossing job} (if it exists). If the crossing job has weight $\le \frac{\eps}{\log(n)} w(OPT(Q))$, we call it \emph{light}, we delete the light crossing job, and we split $\mathcal{R}_{cross}$ at the dividing point $KEY(Q)$. One of these two resulting ranges can inherit the snapping budget of $\mathcal{R}_{cross}$, while we can allocate the other a snapping budget of weight $O(\frac{\eps}{\log(n)} w(OPT(Q))$. We delete/allocate at most $O(\frac{\eps}{\log(n)} w(OPT(Q)))$ weight at every cell, $O(\frac{\eps}{\log(n)} w(OPT))$ weight at every level, and $O(\eps w(OPT))$ weight in total. Also note how all ranges in $RANGES(Q)$ are now completely contained within either $Q_L$ or $Q_R$. Otherwise, if the crossing job has large weight, we call it $\emph{heavy}$ and must find some way to include it in our solution instead of deleting it.

\paragraph{Utilizing heavy crossing jobs.}
We now focus on showing how to construct our solution using a \emph{heavy} crossing job. Our goal is to split $\mathcal{R}_{cross}$ into three parts: one range completely within $Q_L$, some sparse ranges that will be $SPARSE(Q)$ and include the crossing job among other jobs, and one range completely within $Q_R$. As an overview, we will start by considering the smallest time range that contains the crossing job and spans the grid between two (not necessarily consecutive) endpoints in $Z(Q)$.
This range may contain many jobs in $OPT$, so we perform an additional refinement to divide it up into sparse time ranges. In this refinement, we will split up the time range such that we do not delete too much weight and, moreover, all of the resulting time ranges have at most $\nicefrac{1}{\eps}$ jobs. These time ranges now constitute $SPARSE(Q)$. A detailed description of this process of determining $SPARSE(Q)$ is given in stages from ``utilizing heavy jobs'' to ``sparsifying regions.'' For an example of this process that uses the terminology later described in these stages, see \cref{fig:sparse}. 
Any remaining time ranges not selected at this stage will effectively be dense time ranges, and are delegated into $RANGES(Q_L),RANGES(Q_R)$ (after dealing with their alignment issues). This process of designating time ranges to delegate is detailed in stages from ``creating dense ranges'' to ``resolving leafs.''

As a reminder, we have chosen $Z(Q)$ such that the total weight inside any grid slice (a time range between two consecutive endpoints of $Z(Q)$) of $Q$ is at most $\frac{\varepsilon^2}{\log^2(N)} w(P(Q))$.
Recall that $Z(Q)$ contains grid endpoints. For the heavy crossing job, consider the grid endpoint immediately to its left and to its right. Without loss of generality, consider the right one and call it $r$. How we proceed can be split into two cases:

\begin{enumerate}[(1)]
\item In the first case, $r$ overlaps a job $J$ in $OPT'$ with weight at most $ \frac{\eps}{\log(n)} w(OPT(Q))$. We delete $J$ and draw a boundary at $r$. 
In doing this, we will charge the weight of $J$ against the cell $Q$. There are at most two jobs we charge in this manner for that original heavy interval, one for the grid endpoint to the right and one to the left. Meaning, each cell will be charged in this manner at most twice for a total of $O(\frac{\eps}{\log(n)} w(OPT)$ weight at each level and $O(\eps w(OPT))$ weight overall.

\item In the other case, $r$ overlaps a job $J$ that has weight greater than $\frac{\eps}{\log(n)} w(OPT(Q))$. We call $J$ a \emph{highlighted} job. Our algorithm proceeds by considering the grid endpoint immediately to the right of $J$. We determine what to do with this grid endpoint in a recursive manner. Meaning, we proceed in the same two cases that we did when considering what to do with $r$, and continue this recursive process until we finally draw a boundary.
\end{enumerate}

After this process, we will have drawn a \emph{region} (time range corresponding to where we drew a left and right boundary for) in which $OPT'$ has the one heavy crossing job, a number of highlighted jobs (possibly zero), and potentially some remaining jobs that are neither crossing nor highlighted (we call these \emph{useless}). It is our goal to convert this region into time ranges that we can use as sparse time ranges. 
Our process also guarantees this region has borders with endpoints in $Z(Q)$.
Note that we have created a region within some time range of $RANGES(Q)$, but not every point in the time range is necessarily contained within the region.

\paragraph{Deleting useless jobs.} In the generated region, we define \emph{useless} jobs as all jobs that are neither crossing nor highlighted. Useless jobs are completely contained within grid slices. We want to convert the region into sparse time ranges, but there may be many useless jobs that make the region very dense.
Thus, we will delete all jobs in the region that are useless.
By the process of generating the region, any such job is fully contained within a grid slice for which there is a heavy crossing job or highlighted job partially overlapping the grid slice.
We charge deletion of all useless jobs in a given slice by charging against a highlighted or heavy crossing job that must partially overlap the given slice.
By definition of $Z(Q)$, useless jobs in the slice add up to a total weight of at most $\frac{\varepsilon^2}{\log^2(n)} w(OPT(Q))$. This is because we set $Z(Q)$ with $X = \frac{\log^2 (n)}{\eps^2}$ and thus the optimal solution within any grid slice has total weight at most $\frac{\eps^2}{\log^2 (n)} w(OPT(Q))$. Moreover, $\frac{\varepsilon^2}{\log^2(N)} w(OPT(Q))$ is at least a factor of $\eps$ less than the highlighted or heavy crossing job we are charging against (and there are only two such slices whose useless jobs are charging against any highlighted or heavy jobs).

\paragraph{Sparsifying the region.} Now, the region only contains heavy crossing job or highlighted jobs.
We aim to split the region into ranges for $SPARSE(Q)$ without deleting much weight.
The region may have more than $\frac{1}{\varepsilon}$ jobs (meaning it is not sparse). If this is the case, we desire to split the region into time ranges that each have $\le \frac{1}{\eps}$ jobs and start/end at grid endpoints of $Z(Q)$. To do so, we number the jobs in a region from left to right and consider them in groups based on their index modulo $\frac{1}{\varepsilon}$. Note that a group does not consist of consecutive jobs.
Then, we delete the group with lowest weight. We delete this group because we make the observation that all remaining jobs in the region must contain a grid endpoint within it. This is because heavy crossing jobs must contain a grid endpoint by how we defined $Z(Q)$, and highlighted jobs must contain a grid endpoint by their definition. Thus, we can delete the jobs belonging to the lightest group and split the time range at the grid endpoints contained inside each of the deleted jobs.
In doing so, we lose at most a factor of $\eps$ of the total weight of all the considered jobs.
However, now each resulting time range will have at most $\frac{1}{\varepsilon}$ jobs and thus will be a valid sparse range in $SPARSE(Q)$ (because for any range containing a number of consecutive jobs greater than $\frac{1}{\eps}$, we will have split it). Note that all these sparse ranges have endpoints in $Z(Q)$. With all of its terminology now defined, readers may find the example illustrated in \cref{fig:sparse} helpful for their understanding.

\begin{figure}[htbp]
\centerline{\includegraphics[width=0.9\textwidth]{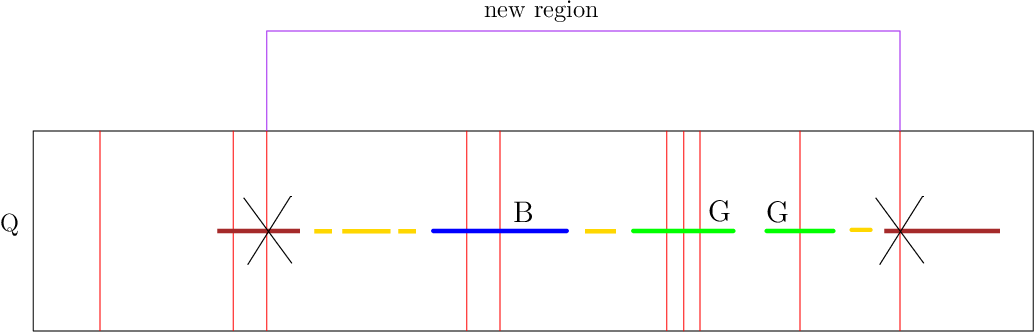}}
\caption{This example illustrates how the sparse regions are created. All vertical segments within $Q$, which are red in the figure, correspond to the points in $Z(Q)$. The cell $Q$ is divided by $Z(Q)$ such that the optimal solution within every grid slice is small. As a reminder, a grid slice is an open time-interval between two consecutive points in $Z(Q)$; see \cref{definition:grid-slice} for a formal definition. We start with the heavy crossing job (the blue horizontal segment marked by ``B''). From this heavy crossing job, we expand the region outwards as necessary. In this example, we expanded to the right, seeing two highlighted jobs (the green horizontal segments marked by ``G'') until we saw a job with low enough weight intersecting a grid endpoint (these job segments are colored in brown and crossed). We delete such brown jobs, and use the grid endpoints they intersected to define the region (outlined in purple and annotated by ``new region''). Useless jobs (pictured in yellow) are then deleted. Later, we sparsify the region.}
\label{fig:sparse}
\end{figure}

\paragraph{Snapping dense ranges.} Recall that not all of the time ranges that we are modifying from $RANGES(Q)$ were part of the region. In particular, there are the time ranges originally in $RANGES(Q)$ other than $\mathcal{R}_{cross}$, as well as the time range in $\mathcal{R}_{cross}$ to the left of the region, and to the right of the region. 
We call these remaining time ranges our \emph{dense ranges} because they may contain many jobs. Note how all dense range are now completely contained within $Q_L$ or $Q_R$.
Ideally, we assign dense ranges to $RANGES(Q_L)$ or $RANGES(Q_R)$. 
However, the remaining dense time ranges have one remaining potential issue, that their endpoints may not align with $Z(Q_{child})$ even though they align with $Z(Q)$. For an example of this issue, see \cref{fig:snap}. The core of this problem is that these dense time ranges correspond to time ranges we would like to delegate to children of $Q$ (i.e., add to $RANGES(Q_L)$ and $RANGES(Q_R)$). However, there is the requirement that time ranges delegated to $RANGES(Q_L)$ and $RANGES(Q_R)$ must have endpoints in $Z(Q_L)$ and $Z(Q_R)$, respectively. Therefore, we have to modify the dense ranges so they align with the grid endpoints of one of $Q$'s children. It is tempting to naively ``snap'' the endpoints of these time ranges inward to the nearest grid endpoints of $Z(Q_{child})$, meaning to slightly contract the endpoints of the time ranges inward so they align with $Z(Q_{child})$. 
Unfortunately, this might result in some jobs being ignored in the process (as illustrated in \cref{fig:snap}); a cell does not consider jobs which are not within a given range. If these ignored jobs have non-negligible total reward, ignoring them can result in a poor solution.
In the stage ``snapping dense ranges'' we detail a more involved contraction-like snapping process that contracts inwards similar to our argument for expanding outwards from heavy crossing jobs when we determined sparse ranges. In our contraction-like snapping process, we convert some of the beginning and end of the dense range into sparse ranges, so we do not need to delete some of the high-reward jobs that we would need to delete with naive snapping. In the stages from ``using essential jobs'' to ``resolving leafs'', we detail how to apply modifications to fulfill the required properties and how to analyze the contraction process with charging arguments.

\begin{figure}[htbp]
\centerline{\includegraphics[width=0.9\textwidth]{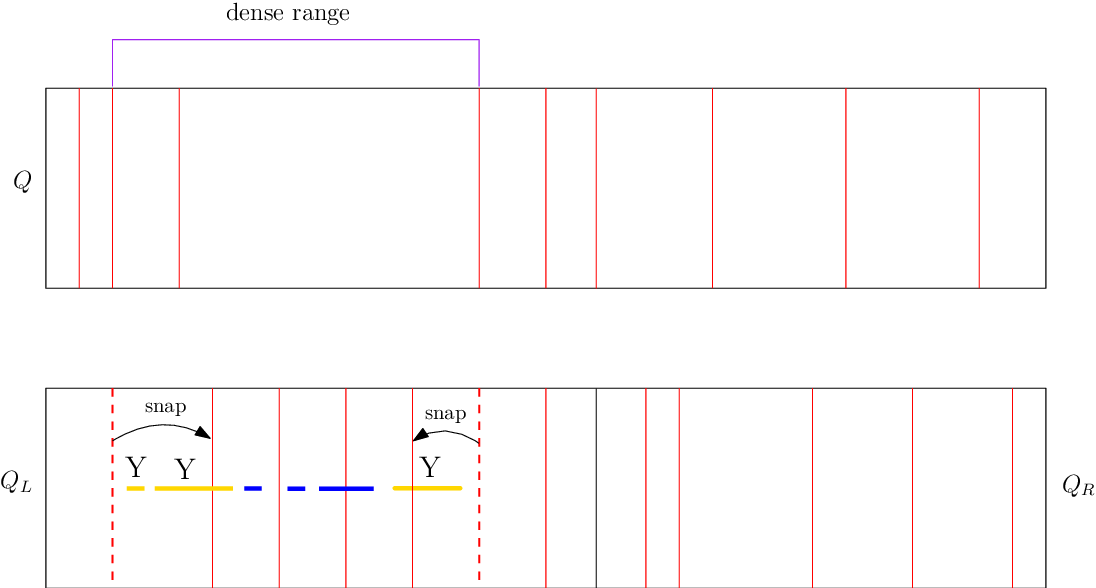}}
\caption{This example illustrates why the snapping we perform has to be done with care. The horizontal segments in this figure represent jobs. We show an initial dense range (outlined in purple) with endpoints in $Z(Q)$. With dashed vertical lines, we show where these endpoints are in $Q_L$. Importantly, they are not aligned with $Z(Q_L)$, i.e., the vertical dashed lines do not belong to $Z(Q_L)$. However, our structure \emph{requires} that dense ranges align with $Z(Q_{child})$, so we must address this. If we were to naively snap the endpoints of the dense range inwards to the endpoints of $Z(Q_L)$, then we would need to delete some jobs (these deleted jobs are colored in yellow and marked by ``Y''), while some other jobs would not be affected (like the remaining jobs in this example, those colored in blue). While this naive snapping may be fine in some cases, it will incur significant loss in cases in which the ``Y'' jobs have large weight. Notice that naively snapping outward to define a new region corresponding to the purple one is not a solution neither, as this could cause the dense time range to overlap with a previously selected sparse time range. Having overlapping ranges can cause us to choose intersecting jobs, and thus an invalid solution. Thus, we detail a more comprehensive manner of dealing with snapping.}
\label{fig:snap}
\end{figure}

Consider an arbitrary unaligned dense time range $U$. Ideally, we would ``snap'' the endpoints of $U$ inward to the nearest grid point of $Z(Q_{child})$ (i.e. move the left endpoint of $U$ to the closest grid point of $Z(Q_{child})$ to its right, and the right endpoint of $U$ to the closest grid endpoint of $Z(Q_{child})$ to its left). However, doing so may force us to delete a job in $OPT'$ that is too valuable (as we would have to delete jobs that overlap the section of $U$ that was snapped inwards). So, we will handle $U$ differently.
Without loss of generality, suppose we want to ``snap'' inward the left endpoint of $U$ to align with $Z(Q_{child})$. Doing so may leave some jobs outside the snapped range. We define the cost of snapping as the total weight of jobs that were previously contained within the range but are no longer completely contained within after snapping.
If immediately snapping inward the left endpoint to the nearest grid point of $Z(Q_{child})$ would cost at most $\frac{2\eps}{\log^2(n)} w(OPT(Q))$, we do that immediately.
Otherwise, this snapping step would cost more than $\frac{2\eps}{\log^2(n)} w(OPT(Q))$, implying that there is a job that overlaps with the grid endpoint of $Z(Q_{child})$ to the right of $U$'s left endpoint (all other jobs we are forced to delete are strictly inside a slice of $Z(Q_{child})$ and thus have total weight $\le \frac{\varepsilon^2}{\log^2(n)} w(OPT(Q_{child})) \le \frac{\varepsilon^2}{\log^2(n)} w(OPT(Q))$) and has weight of at least $\frac{2\eps}{\log^2(n)} w(OPT(Q)) - \frac{\varepsilon^2}{\log^2(n)} w(OPT(Q)) \ge \frac{\eps}{\log^2(n)} w(OPT(Q))$.
We mark that job as ``essential''.
\\
Then, we look to the right of that essential job and examine the job that overlaps the next grid endpoint to the right in $Z(Q_{child})$. If this job has weight at most $\frac{2\varepsilon}{\log^2(n)} w(OPT(Q))$, we delete it and draw a boundary.
Otherwise, we mark it as ``essential'' and continue (following the same process). When we are done, we have a prefix of the dense time range that contains some number of ``essential'' jobs and other jobs, and then a border at a grid endpoint of $Z(Q_{child})$. The final ``snapping'' where we deleted jobs to add the split point had cost $\le \frac{2\varepsilon}{\log^2(n)} w(OPT(Q))$. In essence, these essential jobs are the collection of jobs that were too valuable for us to delete them when we were undergoing the snapping process.

\paragraph{Using essential jobs.}
We will assume this dense time range had a snapping budget and charge the aforementioned final snapping cost to that.
Now, we just need to find a way to use the time range prefix with the essential jobs. We delete all jobs that are not essential in this time range with a similar argument as earlier, that such a job is completely contained in a grid slice with total weight of jobs $\le \frac{\varepsilon^2}{\log^2(N)} w(OPT(Q))$ which is at most a factor of $\varepsilon$ of an essential job partially contained within the slice (and it is partially contained within at most two slices). Then, we convert this time range of essential jobs (with potentially many such essential jobs) into sparse time ranges in the same way as done previously during the ``sparsifying regions'' step. We do so by grouping the jobs according their index modulo $\frac{1}{\varepsilon}$, deleting the group with least total weight, and drawing a border at the grid endpoint of $Z(Q_{child})$ contained within the deleted jobs. Again, by our process we know all such essential jobs must contain a grid endpoint. This creates sparse time ranges with endpoints in $Z(Q)\cup Z(Q_{child})$ and our dense time range has endpoints in $Z(Q_{child})$ so they are both valid.

\paragraph{Financing a snapping budget.} Finally, we need to show that we actually have a sufficient snapping budget. Consider our dense time ranges. We may adjust their endpoints in other scenarios, but we only split dense time ranges into more dense time ranges when they are a crossing range. As only one there is only one crossing range at every cell $Q$, if we give the newly created range a snapping budget of $O(\frac{\varepsilon}{\log(n)} w(OPT(Q)))$, then we do not lose more than $O(\varepsilon w(OPT))$ in total. We showed above that each dense range will use at most $O(\frac{\varepsilon}{\log^2(N)} w(OPT(Q)))$ of its snapping budget at each level, so it will will use $O(\frac{\eps}{\log(n)} w(OPT(Q)))$ in total and stay within its allotted budget of $O(\frac{\eps}{\log(n)} w(OPT(Q)))$ throughout.

\paragraph{Resolving leafs.} Finally, when we have a time range but it cannot be delegated to $Q_{child}$ because $Q_{child}$ does not exist, note there is only possibly room for one job in the range (as by definition of the decomposition of $Q$, no job starts or ends in this range). So we simply consider this range as part of $SPARSE(Q)$.

This now concludes the proof by providing a way to convert $OPT$ to a solution $OPT'$ that obeys our structure and is a $(1+\varepsilon)$-approximation of $OPT$.
\end{proof}

\subsection{Efficiently Approximating Sparse Solutions}
\label{sec:dp-sparse}
Now, we focus on designing an  efficient algorithm for approximating optimal solution in a sparse  time range.

\lemmadpsparse*
\begin{proof}

To prove this claim, we use a dynamic programming approach where our state is the total weight of jobs selected so far. The dynamic programming table $\earliest$ contains for state $X$, $\earliest[X]$, the earliest/leftmost point in time for which the total weight of $X$ is achieved. If we implement this dynamic programming directly, it would require space proportional to the value of solution (which equals the largest possible $X$). Our goal is to avoid this time/space dependence. To that end, we design an \emph{approximate} dynamic program that requires only poly-logarithmic dependence on the value of an optimal solution. We derive the following technical tool to enable this:

\begin{claim}
\label{claim:approximate-sum}
Let $S$ be the set of all powers of $(1+\varepsilon/K)$ not exceeding $W$, i.e., $S = \{(1 + \eps/K)^i \mid 0 \le i \le \lfloor \log_{1 + \eps/K}{W} \rfloor\}$.
Consider an algorithm that supports the addition of any $K$ values (each being at least $1$) where the sum of these $K$ values is guaranteed to not exceed $W$. The values are added one by one. After each addition step, the algorithm maintains a running-total by rounding down the sum of the new value being added and the previous rounded running-total to the nearest value in $S$. Then, the final running-total of the algorithm is a $(1+\varepsilon)$ approximation of the true sum of those $K$ values.
\end{claim}

\begin{proof}
Consider the sequence of $K$ values and thus $K$ additions. Let $OPT$ denote the exact sum of the $K$ values. Let $SOL$ denote the running-total we achieve at the end of our additions. Finally, let $\CUR_i$ denote the running-total as we do these additions at the beginning of stage $i$, which must be in $S$ at the end of every stage. We prove that $SOL \ge (1-\varepsilon)OPT$ and thus $SOL$ is a $(1+\varepsilon)$ approximation of $OPT$. Initially, $\CUR_0=0$. Each step, we add some value $v_i$ to $\CUR_i$. This new value $\CUR'_i = \CUR_i + v_i$. Then, we round $\CUR'_i$ to the nearest power of $(1+\varepsilon/K)$ and denote this as $\CUR''_i$. We call the amount we lose by rounding down the loss $\ell_i = \CUR'_i - \CUR''_i$. For the next stage, we denote $\CUR_{i+1}=\CUR''_i$. Note that 
\[
    \frac{\ell_i}{OPT} \le \frac{\ell_i}{SOL} \le \frac{\ell_i}{\CUR''_i} = \frac{\CUR'_i - \CUR''_i}{\CUR''_i} \le \frac{\varepsilon}{K}
\]
or, otherwise, we would have rounded to a different power of $(1+\varepsilon/K)$. Thus, $\ell_i \le OPT (\frac{\varepsilon}{K})$. Note that $SOL=\CUR_{K}$ and $\CUR_{K} + \sum_{i}\ell_i = OPT$. As such,
\[
    SOL = OPT - \sum_i \ell_i \ge OPT - K \rb{OPT \rb{\frac{\varepsilon}{K}}} = OPT - \varepsilon \cdot OPT = (1 - \varepsilon) OPT.
\]
\end{proof}

Inspired by \cref{claim:approximate-sum}, we now define a set of states $S$ as follows. Our states will represent powers of $(1+\varepsilon/K)$ from $1$ to $Kw$, and hence
\[
    |S| = O\rb{\frac{\log(Kw)}{\log(1+\varepsilon/K)}} = O\rb{\frac{K\log(Kw)}{\varepsilon}}. 
\]

Using this, we create a set of states $S$ which corresponds to powers of $(1+\varepsilon/K)$ from $1$ to $Kw$ (and $0$). We want to maintain for each of these states, approximately the smallest prefix with at most $K$ jobs where we could get total weight approximately equal to $s \in S$. To do this, we loop over the states in increasing order of value. Suppose the current state corresponds to having approximate weight $s \in S$ and $\earliest[s]$ is the shortest prefix we have that has approximate weight $s$. Then we loop over all rounded weights $v \in \{(1+\varepsilon)^i\}$. There are $O(\nicefrac{\log(w)}{\varepsilon})$ such $v$. For each $v$, set $\mathcal{V}$ to be the value of $s+v$ rounded down to the nearest power of $(1+\varepsilon/K)$. Then, if the earliest ending time of a job with rounded weight $v$ that starts after $\earliest[s]$ is less than $\earliest[\mathcal{V}]$, we update $\earliest[\mathcal{V}]$ to that ending time. We can calculate the earliest ending time of any job, with a particular rounded weight, starting after some specified time in $O(\log(n))$ time by maintaining a balanced binary search tree (as done in \cref{section:dynamic-unit}) for each of the $O(\nicefrac{\log(w)}{\varepsilon})$ rounded weights (to powers of $(1+\eps)$). This negligibly adds $O(\log(n))$ time to each update. In total, this solution runs in $O(\frac{K \log(n) \log(w) \log(Kw)}{\eps^2})$ time.

As we can ignore all jobs with weight less than $\nicefrac{\eps}{n} w_{max}$, then we can only focus on jobs with weights in $[\frac{\eps w_{max}}{n},w_{max}]$ and effectively modify $w$ to be $\nicefrac{n}{\eps}$ by dividing all weights by $\frac{\eps w_{max}}{n}$.
This enables us to use $w=O(\nicefrac{n}{\eps})$ in the above runtime bound. As such, this algorithm runs in $O(\frac{K \log(n) \log(n/\eps) \log(K n /\eps)}{\eps^2})$ time.

To show the algorithm's correctness, observe that since we always round down, we will not overestimate the cost. Moreover, we show with \label{claim:approximate-adding} that any set of $K$ additions will be within a factor of $(1+\varepsilon)$ from its true value.

\end{proof}

\begin{corollary}
For our application, we let $K=\frac{1}{\varepsilon}$. As such, we have a $(1+\varepsilon)$-approximation algorithm of the minimum solution with at most $\frac{1}{\varepsilon}$ jobs that runs in time
\[O\rb{\frac{K \log(n) \log(n/\eps) \log(Kn/\eps)}{\eps^2}}=O\rb{\frac{\log(n) \log^2(n/\eps)}{\eps^3}}=O\rb{\frac{\log^3(n)}{\eps^3}}.
\]
\end{corollary}

\subsection{Dynamically Maintaining $Z(Q)$ -- Proof of \cref{lemma:zq-maintain}} \label{subsection:maintaining-z}
Now, we describe how to maintain $Z(Q)$, to intelligently subdivide the cells with guarantees as restated below:
\lemmazqmaintain*
\begin{proof}
Suppose how all jobs are rounded down to powers of $2$. Note how for a cell $Q$, let $w_{max}(Q)$ correspond to the reward of the job with the largest reward contained completely within $Q$. Clearly, $OPT(Q) \ge w_{max}(Q)$.
Moreover, by discarding all jobs with weight less than $\nicefrac{\eps}{n} \cdot w_{max}(Q)$, we discard jobs with total weight at most $\eps \cdot w_{max}(Q) \le \eps \cdot OPT(Q)$.
Accordingly, we focus just on jobs with weights in range $[\nicefrac{\eps}{n} \cdot w_{max}(Q), w_{max}(Q)]$. As these weights have been rounded to powers of $2$, there are only $\lceil \log(\frac{w_{max}(Q)}{\nicefrac{\eps}{n} w_{max}(Q)}) \rceil = O(\log(n/\eps))$ distinct remaining weights. Moreover, we assume that $\nicefrac{1}{\eps}\le n$, as otherwise we can obtain a better algorithm by simply rerunning the classical static algorithm for each update. Altogether, this implies that it suffices to consider $O(\log(n))$ distinct weights.

In our approach, we consider each distinct weight independently, enabling us to consider a $Z^{i}(Q)$ for only jobs with rounded weight $2^i$. That is, $Z^{i}(Q)$ is computed with respect to a set of jobs all having the same weight, which enables us to treat $Z^{i}(Q)$ computation as if it was performed for the unweighted variant. At the end, we let $Z(Q)$ to be the union over the $O(\log(n))$ different $Z^{i}(Q)$, giving us a $Z(Q)$ with our desired guarantees.
This approach is particularly desirable, as we will show how for a particular fixed weight, i.e., we consider the unweighted variant, we can use ideas very similar to those discussed in \cref{section:dynamic-unit} to obtain the $Z^{i}(Q)$. Expanding our scope, for each rounded weight $2^i$, let us maintain a constant-factor approximation of the unweighted problem using the border-based algorithm of \cref{theorem:unweighted-M=1}. In other words, we run the algorithm of \cref{theorem:unweighted-M=1} with $\eps'=O(1)$ such that it has update time $O(\log(n))$ and maintains an $O(1)$-approximation of the unweighted interval scheduling problem. 

Consider $SOL^{i}$ to be \textbf{the set of points} corresponding to the border-based $O(1)$-approximation when only considering jobs of rounded weight $2^i$. In particular, $SOL^{i}$ contains all start/endpoints of the selected jobs by the approximation, as well as all borders. $SOL^{i}[L,R]$ contains all points in $SOL^{i}$ within $[L,R]$. We define $OPT([L,R],i)$ as the optimal \textbf{number of jobs} one can schedule when only considering jobs with rounded weights $2^i$ when only considering jobs fully contained within $[L,R]$. 

\begin{claim} \label{claim:opt-sol}
For all $i,L,R$, it holds that: $OPT([L,R],i) \le |SOL^{i}[L,R]|$
\end{claim}
\begin{proof}
Recall that the border-based approximation algorithm maintains a set of borders and finds the optimal solution within each border chunk. The optimal solution within is calculated by using the greedy earliest-ending algorithm. In general, consider any job $J$. This job $J$ must contain an endpoint of a job in the approximately chosen solution, or it must contain a border. If this were not the case, there are only two possibilities: (i) $J$ is completely contained within a job chosen by the approximate solution, or (ii) $J$ does not intersect any job chosen by the approximate solution. For case (i): this is impossible as the greedy earliest-ending algorithm would not have chosen the job that completely contains $J$. For case (ii): this is impossible because $J$ could be added to the solution within the corresponding border chunk, and this is impossible because the solution within each border chunk is optimal. As each job $J$ must contain a point of $SOL^{i}[L,R]$, it must hold that $OPT([L,R],i) \le |SOL^{i}[L,R]|$.
\end{proof}

Also note a similar bound in the opposing direction:

\begin{claim} \label{claim:sol-ub}
For all $i,L,R$ it holds that: $|SOL^i[L,R]| \le 3 \cdot OPT([L,R],i) + 3$
\end{claim}
\begin{proof}
From $SOL^i[L,R]$, ignore the at most two points corresponding to endpoints of jobs that are only partially within $[L,R]$, and ignore the first remaining point if it corresponds to a border (for a total of ignoring at most 3 points).
Of the remaining points in $SOL^i[L,R]$, they all correspond to endpoints of jobs fully within $[L,R]$, or are a border following such a job. Note how the number of these jobs with points corresponding to them in $SOL^i[L,R]$ must be at most $OPT([L,R],i)$ by definition. Accordingly, we will charge the two points from each job (and its associated border if there is one) to a different job corresponding from $OPT([L,R],i)$, for a total of at most 3 points of $SOL^i[L,R]$ being charged per job in $OPT([L,R],i)$.
\end{proof}
All $SOL^{i}$ can be maintained with update time $O(\log(n))$ because we only update one unweighted $O(1)$-approximation per job insertion or deletion.
We compute each $Z^i(Q)$ for a cell $Q$ corresponding to time range $[L,R]$, by taking $O(X \log(n))$ quantiles of $SOL^{i}[L,R]$. Each of these $Z^i(Q)$ can be achieved with $O(X \log(n))$ walks down a balanced binary search tree, resulting in $O(X \log^2(n))$ time. We define $Z(Q)$ as the union of the $O(\log(n))$ different $Z^i(Q)$. In total, $Z(Q)$ is obtained in $O(X \log^3(n))$ time and $|Z(Q)| = O(X \log^2(n))$.

Finally, the optimal solution within any grid slice, ignoring jobs with weight less than $\eps \cdot w(OPT(Q))$, is upper-bounded by the union of the independent optimal solutions for each rounded weight. Within each grid slice of any $Z^i(Q)$, the optimal solution of jobs using weight $2^i$ is upper-bounded by $O(\frac{2^i |SOL^{i}[l_Q,r_Q]|}{X \log(n)}) = O(\frac{2^i OPT([l_Q,r_Q],i)}{X \log(n)})=O(\frac{w(OPT(Q))}{X \cdot \log(n)})$ following from \cref{claim:opt-sol}, taking $X \log(n)$ quantiles of $SOL^{i}[l_Q,r_Q]$ to form $Z^i(Q)$, and \cref{claim:sol-ub}.
Accordingly, bounding over the $O(\log(n))$ different $Z^i(Q)$, the optimal solution within each grid slice is at most $O(\nicefrac{w(OPT(Q))}{X})$.

\end{proof}

\subsection{Combining All Ingredients -- Proof of \cref{theorem:weighted-dynamic-M=1}}
\label{sec:combining-ingredients}
Now, we put this all together to get a cohesive solution that efficiently calculates an approximately optimal solution of the desired structure. When we handle an insertion/deletion, we make an update to the corresponding balanced binary search tree $T$. Recall that we use a balanced binary search tree such as a red-black tree so that $T$ has depth $O(\log(n))$ and there are $O(\log(n))$ rotations per update. For the $O(\log(n))$ cells $Q$ corresponding to nodes in $T$ affected by rotations, we will recompute aspects of $Q$ such as $Z(Q)$ and all $DP(Q,S)$.
For each such cell $Q$, we will compute a sparse solution corresponding to each segment formed by considering all pairs of grid endpoints $Z(Q) \cup Z(Q_L) \cup Z(Q_R)$ and a dense solution for each segment $S$ formed by pairs of endpoints $Z(Q)$ denoted as $DP(Q,S)$.

To compute all sparse solutions, we use $O(|Z(Q) \cup Z(Q_L) \cup Z(Q_R)|^2)$ calls to our algorithm from \cref{lemma:sparse-sol-approx} resulting in $O(|Z(Q) \cup Z(Q_L) \cup Z(Q_R)|^2 (\frac{\log^3(n)}{\eps^3}) = O(\frac{\log^8(n)}{\eps^4} (\frac{\log^3(n)}{\eps^3})) = O(\frac{\log^{11}(n)}{\eps^7})$ running time. To obtain this complexity, we use the upper-bound $|Z(Q)| = O(X \cdot \log^2(n))$ from \cref{lemma:zq-maintain} and the fact that we let $X = \tfrac{\log^2 n}{\eps^2}$ in the beginning of \cref{sec:existence-of-nearly-optimal-sparse-structure}.

To compute all $DP(Q,S)$, we build on the proof of \cref{lemma:sol-structure}. Namely, from the proof of \cref{lemma:sol-structure} a $(1+\eps)$-approximate solution is maintained by dividing $S$ into sparse, i.e., $SPARSE(Q)$, and dense segments of $Q_L$ and $Q_R$, i.e., $RANGES(Q_L)$ and $RANGES(Q_R)$.
We update our data structure from bottom to top. Hence, when we update $DP(Q_L)$ and $DP(Q_R)$ it enables us to learn approximate optimal values gained from a set $RANGES(Q_L)$ and $RANGES(Q_R)$. Thus, to calculate $DP(Q,S)$ we consider an interval scheduling instance where jobs start at a grid endpoint of $S$ and end at a grid endpoint of $S$.
In this instance, jobs correspond to all the sparse segments of $Z(Q),Z(Q_L),Z(Q_R)$ and all the dense segments of $Z(Q_L),Z(Q_R)$. 
We compute this dense segment answer for all dense segments $Z(Q)$ in $O(|Z(Q) \cup Z(Q_L) \cup Z(Q_R)|^3)=O\rb{\frac{\log^{12}(n)}{\eps^{6}}}$ time, with a dynamic program where the state is the starting and ending point of a segment and the transition tries all potential grid endpoints to split the range at (or just uses the interval from the start to the end). For each update, we update $O(\log(n))$ cells affected by rotations by recomputing the optimal sparse solutions for segments and the respective $DP(Q,S)$. Finally, at the beginning of each update, we use $O(\log(n))$ calls to our algorithm for computing $Z(Q)$ from \cref{subsection:maintaining-z} with $X = \frac{\log^2(n)}{\eps^2}$ in $O(X \cdot \log^3(n))$ time for $O(\frac{\log^5(n)}{\eps^2})$ time for each cell. As such, our total update time is $O\rb{\log(n) \cdot (\frac{\log^{11}(n)}{\eps^7} + \frac{\log^{12}(n)}{\eps^{6}} + \frac{\log^6(n)}{\eps^2})} = O\rb{\frac{\log^{12}(n)}{\eps^7} + \frac{\log^{13}(n)}{\eps^{6}}}$.

\section{LCA for Interval Scheduling on a Single Machine}\label{section:local}
In this section we design local algorithms for interval scheduling, using techniques developed in \cref{section:dynamic-unit}. 
While our previous algorithm is desirable in that it gives an efficient and simple algorithm to efficiently partition the time dimension and maintain an approximate solution, it requires bookkeeping (our partitioning relies on the past history of requests made before the update). We design local algorithms for interval scheduling that do not require knowledge of such bookkeeping. We need some way to probe information about ``similar'' intervals: as such, we will assume probe-access to an oracle that gives information about other intervals. 
In contrast to the dynamic setting, our oracle has no dependence on $\eps$ and thus can be used for any $\eps$. 
An LCA in this setting will answer queries of the form $LCA(S,\varepsilon)$, where we are given a set of intervals $S$ and approximation parameter $\varepsilon$, and on {\em query interval $I \in S$,  we must answer 
 whether $I$ is in our $(1+\varepsilon)$-approximation} in such a way that is consistent
with all other answers to  queries to the $LCA$ (with the same $\eps$). 
Generally, we develop a partitioning method that does not require much bookkeeping while attempting this, such that our notion of leveraging locality extends beyond any particular computation model.
While achieving this, we assume our LCA is given probe-access to a \emph{successor oracle} that answers what we will call \emph{successor probes} or $successor(x)$: ``\emph{What is the interval with the earliest endpoint, of those that start after point $x$?}'' This is a natural question for obtaining information about local intervals in this setting. 
In particular, given that an interval is in our solution and ends at some point $x$, $successor(x)$ would be the next interval chosen by the classic greedy algorithm (for the unit-weight setting).
Such an oracle could be implemented with $O(\log n)$ time updates and queries (in a manner similar to how $T_{\rm{all}}$ is used in \cref{theorem:unweighted-M=1}). 
Since our LCA outputs different solutions for different choices of $\eps$, there is a strong sense in which an oracle that is independent of $\eps$ (such as the oracle we utilize) is unable to maintain nontrivial bookkeeping (meaning the oracle could not give the LCA nontrivial information about the solution).
We focus on the unit-reward interval \MaxIS problem ($M=1$). Our emphasis in doing so is not the specific problem instance or probe-model (e.g., in \cref{section:scheduling} we modify our probe-model and show an algorithm that works for multiple machines), but instead emphasizing a method of partitioning over time that utilizes locality and limited bookkeeping. %

At a high level, our novel partitioning method can be viewed as a rule-based approach that uses few probes to identify whether any given interval is in our solution. This approach is oblivious to query order. 

To illustrate how to employ successor probes, we will first design a probe-based algorithm, denoted by \ProbeBased. Then, we will describe an exact global algorithm.
We will modify this (exact) global algorithm to an approximate global one by  partitioning time into independent regions that enable a sense of locality. Finally, we will introduce an LCA motivated by the approximate global algorithm.

As mentioned before, suppose we have access to the \emph{successor probe} or $successor(x)$: ``What is the interval with the earliest endpoint, of those that start after point $x$?''
Note that access to such a probe can be provided in $O(\log(n))$ update and probe time. Due to the limited capabilities of probing without bookkeeping, our LCA results will require assuming a bounded coordinate system as was requited for results in the prior work of \cite{henzinger2020dynamic}. In particular, our LCA results assume that all jobs have length at least $1$ and the starting/ending times are withing $[0,N]$.
\begin{lemma}\label{lemma:query-based-opt}
There exists an algorithm (that we call \ProbeBased) that gives an optimal unweighted solution within some range $[L,R]$ with $|OPT|+1$ successor probes.
\end{lemma}

\begin{proof}
We now describe \ProbeBased algorithm.

It is a classic result that an optimal solution for unweighted interval \MaxIS is achieved by greedily choosing the interval with earliest ending point among those that start after the last chosen ending point.
We use such probes to easily simulate a greedy algorithm for the optimal solution within range $[L,R]$. We start by making a successor probe $successor(L)$. If this interval has an ending point at most $R$, we let that interval be the first one in our optimal solution. Otherwise, the optimal solution is of size zero. Now, we calculate the optimal solution within the range $[ending\_point(successor(L)), R]$ in the same way. Thus, we repeatedly make successor probes at the ending point of the last interval we have chosen.
\end{proof}

Moving forward, we prove an LCA in the unit-reward \MaxIS setting:

\theoremunweightedMlocal*
\begin{proof}

\textbf{Hierarchically Simulating Greedy.}
We aim to hierarchically simulate the greedy algorithm so that it will be easier to adapt towards an LCA. To do this, we utilize a binary tree over the integer points in $[0,N]$. For a node $Q$ in this binary tree, its left child is denoted by $Q_L$ and right child denoted by $Q_R$. We use $Q_{mid}$ to denote the midpoint of the interval corresponding to $Q$. The intervals corresponding to $Q_L$ and $Q_R$ are such that they divide $Q$ exactly in half at its midpoint $Q_{mid}$.
We say that an interval $J$ is assigned to the node $Q$ in the binary tree where $J$ starts in the range contained by $Q_L$ and ends in the range contained by $Q_R$. An equivalent characterization is that $J$ is assigned to the \emph{largest node} $Q$, i.e., $Q$ corresponding to the largest interval, where $J$ contains $Q_{mid}$. As all intervals assigned to a node $Q$ share a common point $Q_{mid}$, at most one of them can be in our solution. In our hierarchical simulation, we decide at the node $Q$ which (if any) of the intervals assigned to it will be in our solution. To accomplish this, we define $f(Q,earliest)$ as a function that computes the interval scheduling problem within the range covered by $Q$, assuming we cannot use any interval that starts before the time $earliest$. Our function $f$ will decide which intervals are in our solution, and it will return the end of the last interval chosen in the range covered by $Q$. As such, calling $f(Q_{\rm{root}},0)$ corresponds to calculating the global solution.

\paragraph{Description of \cref{alg:alg-global-exact}.}
We now provide an algorithm for globally computing $f(Q,earliest)$ as \cref{alg:alg-global-exact}. This algorithm simulates the classic greedy approach for calculating the exact unweighted interval \MaxIS. 
Intuitively, this approach proposes a new way of visualizing and computing this greedy process that will be helpful for obtaining fast LCA. 
We simulate the greedy on intervals in $Q_L$ to find the last ending time it will select before $Q_{mid}$, then we determine if the greedy chooses an interval $I_{mid}$ that contains $Q_{mid}$, and finally we simulate the greedy on intervals within $Q_R$.

\begin{algorithm}[h]
	\SetKw{Print}{Print}
    \Input{$Q$ : a tree node, corresponding to a time-range \\
        $earliest$ : earliest starting time for future intervals}
    \Output{Finds/prints a set of non-overlapping intervals such that (1) each interval is contained in $Q$, and (2) no interval starts before $earliest$\\
    Returns ending time of last interval selected so far}
    \BlankLine
    \BlankLine
	
	$after\_left\_earliest \gets f(Q_L, earliest)$ \\
	$I_{mid} \gets$ interval after $after\_left\_earliest$ containing $Q_{mid}$ with earliest end time
	
	\uIf{$I_{mid} \ne \emptyset$ \textbf{and} no interval is contained within $I_{mid}$}{
	    $after\_mid\_earliest \gets end(I_{mid})$ \\
	    \Print $I_{mid}$
	}
	\Else{
	    $after\_mid\_earliest \gets after\_left\_earliest$
	}
	$after\_right\_earliest \gets f(Q_R,after\_mid\_earliest)$ \\
	\Return $after\_right\_earliest$

  \caption{Global, exact algorithm for
$f(Q,earliest)$
    \label{alg:alg-global-exact}}
\end{algorithm}

\begin{lemma}\label{lemma:global_exact}
\cref{alg:alg-global-exact} is a global algorithm for calculating unweighted interval \MaxIS.
\end{lemma}
\begin{proof}
As this algorithm simulates the classic greedy approach, its correctness follows immediately.
\end{proof}
\paragraph{An easier to locally simulate, approximate global process}
We now modify \cref{alg:alg-global-exact} to more easily lend itself to local computation, while weakening our claim from an exact solution to $(1+\eps)$ approximation. This modified global process will serve as an approximate solution that is easier for an LCA to simulate. We first describe the main intuition behind our modification, and then provide more details on how to design the algorithm (see \cref{alg:alg-global-approx}).

Consider a node $Q$ (defined as in \cref{alg:alg-global-exact}) and its left and right children $Q_L$ and $Q_R$, respectively. If optimal solutions within $Q_L$ and $Q_R$ are both large, i.e., have size at least $\nicefrac{1}{\eps}$, we can afford to create a boundary at $Q_{mid}$ and not use any interval containing $Q_{mid}$ (in which case we reduce the size of an optimum solution by at most one), ``charge'' the potential interval intersecting this boundary point $Q_{mid}$ to the size of solutions in $Q_L$ and $Q_R$, and handle $Q_L$ and $Q_R$ independently. \cref{lemma:invariant-value} implies that this approach leads to $(1+\eps)$-approximate scheduling. Being able to handle $Q_L$ and $Q_R$ independently is crucial for designing our desired LCA -- it enables us to explore only one of the two nodes to answer whether a given interval $I$ belongs to an approximate solution or not. Notice that if we have not discarded intervals containing $Q_{mid}$ and if $I$ belongs to the range defined by $Q_R$, then we would need to learn an approximate solution of $Q_L$ first before we could decide whether $I$ is an approximate solution of $Q_R$.

Otherwise, at least one of optimum solutions in $Q_L$ and $Q_R$ contains at most $\frac{1}{\varepsilon}$ intervals. For cells that have at most $\frac{1}{\varepsilon}$ intervals we use \ProbeBased to compute their optimum with $O(1/\eps)$ successor probes. On the node (if any) that has solution larger than $1/\eps$ we simply recurse. As we show in \cref{lemma:alg-local-approx}, this recursion is efficient enough even in the context of LCA. We now provide more details on the algorithm itself.

\paragraph{Description of \cref{alg:alg-global-approx}.}
We now define an algorithm (\cref{alg:alg-global-approx}) for globally computing an approximation of $f(Q,earliest)$. As the first step of the algorithm, we invoke \ProbeBased to identify whether or not simulating the greedy within $Q_L$ and $Q_R$ will both have large solutions with at least $\nicefrac{1}{\eps}$ intervals. (Notice that to obtain this information we do not need to compute the entire solution in $Q_L$ or $Q_R$, but only up to $1/\eps$ many intervals.)
If \emph{both} have solutions of size at least $1/\eps$, the algorithms draws a border at $Q_{mid}$ (hence ignoring any interval that intersects $Q_{mid}$) and simulates the approximate greedy on $Q_L$ and $Q_R$ independently by invoking \cref{alg:alg-global-approx} on $Q_L$ and $Q_R$.

Otherwise, at least one of $Q_L$ and $Q_R$ has an optimal solution of size less than $1/\eps$. \cref{alg:alg-global-approx} simulates exact greedy on nodes that have an optimal solution of size at most $1/\eps$, and invokes \cref{alg:alg-global-approx} recursively on the node (if any) that has larger solution. In addition, \cref{alg:alg-global-approx} determines whether the greedy chooses an interval $I_{mid}$ that contains $Q_{mid}$, which is used to determine parameter $earliest$ for the processing of $Q_R$.

\begin{algorithm}[h]
	\SetKw{Print}{Print}

    \Input{$Q$ : cell \\
        $earliest$ : earliest valid starting time for future intervals \\
		$\eps$ : approximation parameter}
    \Output{Returns ending time of last interval selected so far \\
    Prints each interval in the solution exactly once}
    \BlankLine
    \BlankLine

	\If{$OPT(Q_L) > \nicefrac{1}{\eps}$ and $OPT(Q_R) > \nicefrac{1}{\eps}$}{
	    Draw a border at $Q_{mid}$ \\
	    \tcc{In our LCA, we will need to invoke only one of these.}
	    Invoke $f(Q_L,earliest)$ and $f(Q_R,Q_{mid})$ \\
	    \Return $f(Q_R,Q_{mid})$
	}
	
	\uIf{$OPT(Q_L) \le \nicefrac{1}{\eps}$}{
	    \tcc{See \cref{lemma:query-based-opt} to recall \ProbeBased.}
        $after\_left\_earliest \gets \ProbeBased(Q_L,earliest)$ \\
        \Print intervals in $\ProbeBased(Q_L,earliest)$
	}
	\Else{
	    $after\_left\_earliest \gets f(Q_L,earliest)$
	}
	
	$I_{mid} \gets$ interval after $after\_left\_earliest$ containing $Q_{mid}$ with earliest end time
	
	\uIf{$I_{mid} \ne \emptyset$ \textbf{and} no interval is contained within $I_{mid}$}{
	    $after\_mid\_earliest \gets end(I_{mid})$ \\
	    \Print $I_{mid}$
	}
	\Else{
	    $after\_mid\_earliest \gets after\_left\_earliest$
	}
	
	\uIf{$OPT(Q_R) \le \nicefrac{1}{\eps}$}{
        $after\_right\_earliest \gets \ProbeBased(Q_R,after\_mid\_earliest)$ \\
        \Print intervals in $\ProbeBased(Q_R,after\_mid\_earliest)$
	}
	\Else{
	    $after\_right\_earliest \gets f(Q_R,after\_mid\_earliest)$
	}
	\Return $after\_right\_earliest$

  \caption{Global, approximate algorithm for
$f(Q,earliest)$
    \label{alg:alg-global-approx}}
\end{algorithm}

\begin{lemma}\label{lemma:global_approx}
\cref{alg:alg-global-approx} is a global algorithm for calculating a $(1+\varepsilon)$-approximation of unweighted interval \MaxIS.
\end{lemma}

\begin{proof}

Note that \cref{alg:alg-global-approx} will compute $f(Q,earliest)$ exactly (by simulating the classic greedy) other than when it draws borders so that it can compute answers for $Q_L$ and $Q_R$ independently. However, we only draw borders when both the region the left and right of the border has a solution with at least $\nicefrac{1}{\varepsilon}$ intervals. As such, we maintain the requirements for \cref{lemma:invariant-value} to hold and can simulate the greedy exactly within borders which immediately shows correctness for a $(1+\varepsilon)$-approximation.
\end{proof}

\paragraph{Designing an LCA}
We design an LCA that simulates the approximate, global process of \cref{alg:alg-global-approx}. Note that \cref{alg:alg-global-approx} never recurses on both $Q_L$ and $Q_R$ unless we drew a border between them, in which case the recursive calls are independent. Since we are now designing an LCA that only determines whether a particular interval is in a solution, we can ignore one of the two independent subproblems. So, we design an LCA that only needs to recurse down one child each time and has desirable runtime. We design an algorithm for a slightly modified function $f(Q,earliest,I)$, where we compute whether $I$ is in our solution.

\paragraph{Description of LCA \cref{alg:alg-local-approx}. }
We now define an algorithm for locally computing an approximation of $f(Q,earliest,I)$ in \cref{alg:alg-local-approx}. This algorithm directly builds on \cref{alg:alg-global-approx}, whose description is provided above. The key difference between \cref{alg:alg-local-approx} and \cref{alg:alg-global-approx} is that when \cref{alg:alg-local-approx} draws a border, the algorithm does not calculate both $f(Q_L, earliest, I)$ and $f(Q_R, Q_{mid}, I)$, as they are independent and it suffices to compute the output of only one of those $Q_L$ and $Q_R$. If $I \in Q_L$, then \cref{alg:alg-local-approx} invokes $f(Q_L, earliest, I)$ as the output is independent of $f(Q_R, Q_{mid}, I)$. Otherwise, if $I \in Q_R$ or $I \notin (Q_L \cup Q_R)$, then the algorithm invokes $f(Q_R, Q_{mid}, I)$ as either $I$ has already been decided on whether it will be in the output or the algorithm only needs the result of $f(Q_R, Q_{mid}, I)$. As we show in the next claim, this suffices to guarantee LCA complexity of $O(\log{n} / \eps)$. The rest of algorithm \cref{alg:alg-local-approx} is the same as \cref{alg:alg-global-approx}.

\begin{algorithm}[h]
	\SetKw{Print}{Print}

    \Input{$Q$ : cell \\
        $earliest$ : earliest valid starting time for future intervals \\
		$\eps$ : approximation parameter\\
		$I$ : interval}
    \Output{Returns ending time of last interval selected so far \\
    Prints ``Yes'' once if $I$ is in the desired solution within $Q$, else prints nothing}
    \BlankLine
    \BlankLine

	\If{$OPT(Q_L) > \nicefrac{1}{\eps}$ and $OPT(Q_R) > \nicefrac{1}{\eps}$}{
	    Draw a border at $Q_{mid}$ \\
	    \lIf{$I \in Q_L$} {
	        \Return $f(Q_L, earliest, I)$
	    }
	    \lElse {
	        \Return $f(Q_R, Q_{mid}, I)$
	    }
	}
	
	\uIf{$OPT(Q_L) \le \nicefrac{1}{\eps}$}{
	    \tcc{See \cref{lemma:query-based-opt} to recall \ProbeBased.}
        $after\_left\_earliest \gets \ProbeBased(Q_L,earliest)$
        
        \lIf{$I \in \ProbeBased(Q_L, earliest) $ solution} {
            \Print ``Yes''
        }
	}
	\Else{
	    $after\_left\_earliest \gets f(Q_L,earliest)$
	}
	
	$I_{mid} \gets$ interval after $after\_left\_earliest$ containing $Q_{mid}$ with earliest end time
	
	\uIf{$I_{mid} \ne \emptyset$ \textbf{and} no interval is contained within $I_{mid}$}{
	    $after\_mid\_earliest \gets end(I_{mid})$
	    
	    \lIf{$I_{mid} = I$} {
	        \Print ``Yes''
	    }
	}
	\Else{
	    $after\_mid\_earliest \gets after\_left\_earliest$
	}
	\uIf{$OPT(Q_R) \le \nicefrac{1}{\eps}$}{
        $after\_right\_earliest \gets \ProbeBased(Q_R,after\_mid\_earliest)$ \\
        \lIf{$I \in \ProbeBased(Q_L, after\_mid\_earliest) $ solution} {
            \Print ``Yes''
        }
	}
	\Else{
	    $after\_right\_earliest \gets f(Q_R,after\_mid\_earliest)$
	}
	\Return $after\_right\_earliest$

\caption{Local, approximate algorithm for
$f(Q,earliest)$
    \label{alg:alg-local-approx}}
\end{algorithm}

\begin{lemma}\label{lemma:alg-local-approx}
\cref{alg:alg-local-approx} is a $(1+\varepsilon)$-approximation LCA for unweighted interval \MaxIS using $O(\frac{\log{N}}{\varepsilon})$ successor probes.
\end{lemma}
\begin{proof}

Correctness follows from that our algorithm simulates \cref{alg:alg-global-approx} which is a $(1+\eps)$-approximation by \cref{lemma:global_approx}. To show that our LCA is efficient, we note that at each of the $\log(N)$ levels we only invoke one instance of $f$ for a child. Additionally, we only use $O(\frac{1}{\varepsilon})$ successor probes at each of these levels. We identify when $OPT(Q_L)$ and $OPT(Q_R)$ are greater than $\nicefrac{1}{\eps}$ by using $\nicefrac{1}{\eps}+1$ steps of \ProbeBased. So in total, our LCA only uses $O(\frac{\log(N)}{\varepsilon})$ successor probes.
\end{proof}
Thus, we have our desired LCA.
\end{proof}

Such an approach can use other probe-models that enable us to effectively simulate successor probes. For example, we could consider a probe-model where we want to know all intervals that intersect a certain point. Regardless, our goal is to emphasize this partitioning method that enables more local algorithms due to its lack of bookkeeping. 
\section{Scheduling Algorithms on Multiple Machines with Partitioning}\label{section:scheduling}
In the previous sections we focused on the case of a single machine, i.e., $M = 1$. In this section, we extend our results to the setting where there are multiple machines on which to schedule jobs ($M > 1$).
In particular, we obtain the following results
\begin{restatable}[Unweighted dynamic, multiple machines]{theorem}{theoremunweighteddynamicschedule}
\label{theorem:unweighted-dynamic-scheduling-multiple}
Let $\cJ$ be a set of $n$ jobs. For any $\eps > 0$, there exists a fully dynamic algorithm for $(1 + \eps)$-approximate unweighted interval scheduling for $\cJ$ on $M$ machines performing updates in $O\rb{\frac{M \log(n)}{\eps}}$ and queries in $O(\log(n))$ worst-case time.
\end{restatable}

\begin{restatable}[Unweighted LCA, multiple machines]{theorem}{theoremunweighteddynamicscheduling}
\label{theorem:unweighted-LCA-scheduling}
Let $\cJ$ be a set of $n$ jobs with their ending times upper-bounded by $N$. For any $\eps > 0$, there exists a local computation algorithm for $(1+\eps)$-approximate unweighted interval scheduling for $\cJ$ on $M$ machines using $O\rb{\frac{M\log(N)}{\eps}}$ probes. 
\end{restatable}

\begin{restatable}[Weighted dynamic, multiple machines]{theorem}
{theoremweighteddynamicscheduling}
\label{theorem:weighted-dynamic-scheduling}
Let $\cJ$ be a set of $n$ jobs. For any $\eps > 0$, there exists a fully dynamic algorithm for $\rb{\frac{M^M}{M^M-(M-1)^M}(1 + \eps)}$-approximate\footnote{Note that this goes to $\frac{e}{e-1}(1+\eps) \approx 1.58 (1 + \eps)$ from below as $M$ tends to $\infty$.} weighted interval scheduling for $\cJ$ on $M$ machines performing updates in $O\rb{\frac{Mw \log(w) \log(n)}{\eps^3}}$ and queries in $O(\log(n))$ worst-case time.
\end{restatable}

We provide general reductions that show how to reduce interval scheduling on multiple machines to the same task on a single machine. Our reductions incur only a small constant factor loss in the approximation and are easy to simulate in the dynamic setting.
\begin{restatable}{theorem}{theoremrandomunweighted}
\label{theorem:random-unweighted}
Given an oracle for computing an $\alpha$-approximate unweighted interval scheduling on a single machine, there exists a randomized algorithm for the same task on $M$ machines that yields an $(2-1/M)\alpha$-approximation in expectation.
\end{restatable}

\begin{restatable}{theorem}{theoremweightedrandom}
\label{theorem:weighted-random}
Given an oracle for computing an $\alpha$-approximate weighted interval scheduling on a single machine, there exists a randomized algorithm for the same task on $M$ machines that yields an $e \cdot \alpha$ approximation in expectation.
\end{restatable}

Note that the approximation guarantees we obtain in \cref{theorem:unweighted-dynamic-scheduling-multiple,theorem:weighted-dynamic-scheduling} are stronger than a direct application of \cref{theorem:random-unweighted,theorem:weighted-random} on \cref{theorem:unweighted-M=1,theorem:weighted-dynamic-M=1}. However, our reductions importantly give rise to significantly faster dynamic algorithms for scheduling on multiple machines, having no dependence on $M$. Concretely, \cref{theorem:random-unweighted,theorem:weighted-random} result in algorithms with the same time complexity as \cref{theorem:unweighted-M=1,theorem:weighted-dynamic-M=1} and only an increase in expected approximation factor of $(2-1/M)$ and $e$, respectively. The same running time is obtained because \cref{theorem:random-unweighted,theorem:weighted-random} assign jobs to machines in negligible time, and then each update or query just results in an update or query on the corresponding data structure for one machine.

\subsection{Overview}
Now we detail our techniques for extending interval scheduling methods for one machine, to scheduling for many machines. A key difficulty in extending such methods is that there is an \emph{inherent dependency} in the process of scheduling. Choosing to use (or to not use) a job on one machine directly affects the optimal schedule for the remaining machines. To overcome this, our work examines two approaches for scheduling. With the first approach, we maintain approximation guarantees almost the same to those for a single-machine setting at the expense of an $O(M)$ factor slowdown. With the second approach, we achieve the same time complexity as was achieved for a single machine at the expense of a slight multiplicative decrease in approximation guarantees.

\textbf{Partitioning over machines and time \emph{simultaneously}.} First, we explore partitioning over time and machines \emph{simultaneously}. At a high level, we do so by dynamically maintaining a partition over time and computing an approximately optimal solution for all machines together within each time range. However, as computing a solution for machines together is a process with dependencies, our algorithm incurs at least an $O(M)$ factor slowdown compared to analogous approaches for a single machine. 

\emph{Unweighted jobs.} For scheduling unweighted jobs on multiple machines, there is a well-known centralized greedy approach similar in style to the greedy for scheduling unweighted jobs on one machine. As this greedy is efficient to simulate, we can employ an algorithm and analysis similar to how we dynamically computed unweighted interval scheduling on one machine. The notable difference is that we might need to charge $M$ jobs containing a border against our solutions in adjacent regions (as opposed to just charging $1$). Accordingly, we maintain borders where the optimal solution inside each region is size $\Theta(\frac{M}{\eps})$. 

\emph{Weighted jobs.} Using similar approaches in the setting with weighted jobs faces to challenges we must overcome. First, the well-known approach for computing this problem in the centralized setting uses minimum-cost flow (as opposed to a greedy) which is not clear how to efficiently simulate dynamically. To handle this, within borders we instead compute the weighted maximum independent set $M$ times which will only lose a factor of $\frac{M^M}{M^M - (M-1)^M}$ (upper-bounded by approximately $1.58$) in approximation guarantee. To compute the weighted maximum independent set, we use a dynamic programming approach. Finally, we note that we might need to charge $M$ jobs of reward $w$ containing a border ($Mw$ total reward) against our solutions in adjacent regions. So, we maintain borders where the optimal solution inside each region has total reward $\Theta(\frac{Mw}{\eps})$.

\textbf{Partitioning over machines \emph{then} time.} Our second approach avoids the slowdown of the first, at the expense of a small multiplicative decrease in the approximation guarantee. To do so, we first partition jobs over machines and \emph{then} dynamically partition time to maintain solutions for each machine independently. In both of our results, we partition jobs among machines by assigning each job to a machine uniformly at random. Then, for each machine we simply maintain an approximately optimal schedule among jobs that were randomly assigned to it. This reduction immediately yields algorithms that are asymptotically just as fast as scheduling with only one machine. We now outline our techniques for showing this approach still maintains a strong approximation guarantee:

\emph{Unweighted jobs.} For scheduling on multiple machines, we note a symmetry among machines. If we can calculate the expected optimal solution of jobs assigned to a particular machine, then the expected optimal solution after all job assignments is simply this quantity multiplied by $M$.
To show a lower-bound for the expected optimal solution among jobs assigned to a particular machine, we recall that unweighted interval scheduling on one machine can be solved with a simple greedy method where we consider jobs in increasing order of their ending time and include the job if it does not intersect any previously included jobs. Interestingly, our method simulates this greedy on one machine by considering all jobs in an optimal solution for $M$ machines, where we lazily do not yet realize whether or not each job was assigned to this particular machine.
Then, as we run our greedy, we realize whether or not a job is assigned to this particular machine only when the greedy would choose to include this job. If we realize that this job was not assigned to this machine, then we continue the greedy method accordingly.
Otherwise, we continue the greedy as if we included this job, and we delete the at most $M-1$ jobs with later ending times that intersect this job (we obtain this $M-1$ upper-bound because we know the set of all jobs forms a valid solution on $M$ machines).
Whether or not a particular job is assigned to this particular machine is a Bernoulli random variable with parameter $\frac{1}{M}$, and we thus expect to see $M$ jobs that our greedy would select until we can actually use one on this machine. In total, our expected proportion of used jobs (among those in a particular optimal solution on $M$ machines) for this machine is at least $\frac{1}{M + (M-1)}$, so our global solution only loses a factor of $(2-1/M)$ in expectation by randomly assigning jobs to machines.

\emph{Weighted jobs.} The weighted setting presents unique challenges that the unweighted setting does not. For example, in our greedy-simulation approach for analyzing the reduction in the unweighted setting, one can show how long jobs that contain many other jobs are less likely to be included in the obtained solution for any machine. This is because, when we delete the $M-1$ jobs with later ending times than some particular job we chose to include in our solution, this will often delete the longer job that contains many jobs.
This is problematic in the weighted setting, as the longer job may provide extremely large reward. To handle this, we provide a different analysis where every job in some optimal solution among $M$ machines, has at least constant probability of being in our solution after randomly assigning jobs to machines.
To accomplish this, we introduce the following procedure. First, generate a uniformly random permutation and process all jobs (from the particular optimal solution on $M$ machines) in this order. When we process a job $J$, we include it on its assigned machine's schedule if (i) there are no jobs intersecting $J$ that are currently in $J$'s assigned machine's schedule, or (ii) all jobs intersecting $J$ that are currently in $J$'s assigned machine's schedule are completely contained within $J$. Note that if $J$ is selected because of the latter criteria, we delete all jobs scheduled in its assigned machine that are completely contained within $J$. With a detailed analysis, we show that no matter what the original optimal schedule over $M$ machines is, each job has probability at least $\frac{1}{e}$ of being included in the final schedule by using our procedure.
So, our global solution only loses a factor of $e$ in expectation by randomly assigning jobs to machines.

\subsection{Unweighted Interval Scheduling on Multiple Machines}

An efficient centralized/sequential algorithm to exactly calculate unweighted interval scheduling has structure very similar to the greedy algorithm for unweighted interval \MaxIS. We use that to show that modifications of our results for single-machine setting lead to results in the multiple-machine setup.

\theoremunweighteddynamicschedule*

For the setting of local unweighted interval scheduling, we show the following.
\theoremunweighteddynamicscheduling*
These theorems are proved in \cref{proof:theorem:unweighted-dynamic-scheduling-multiple,proof:theorem:unweighted-LCA-scheduling}.

\subsection{Weighted Interval Scheduling on Multiple Machines}
For the weighted interval scheduling problem, the well-known minimum-cost flow based algorithm requires $O(n^2 \log(n))$ time. It is not clear how to efficiently simulate this approach in the dynamic or local setting. Instead, we consider alternative approaches for partitioning jobs over machines.
When $M=1$ for scheduling, the optimal solution has a structure similar to that of \MaxIS. \cite{bar2001approximating} study a natural greedy approach for $M > 1$ which consists of $M$ times performing the following: in the $i$-th step take the (weighted) \MaxIS of the currently non-scheduled jobs; schedule these jobs on the machine $i$.
(To be precise, we note that \cite{bar2001approximating} study this algorithm in a more general variant of weighted interval scheduling where start/end times are flexible.)) Theorem 3.3 of \cite{bar2001approximating}  implies that using an $\alpha$-approximation for \MaxIS $M$ times, in the way as described above, gives a $\frac{ (\alpha M)^M} { (\alpha M)^M - (\alpha M - 1)^M}$-approximation (and thus a $\frac{\alpha M^M}{M^M - (M-1)^M}$ approximation) for weighted interval scheduling. Hence, a natural question to ask is whether this approximation can be retained even when using approximate algorithms and in settings other than centralized. We answer this question affirmatively by showing the following results, whose proof is deferred to \cref{proof:theorem:weighted-dynamic-scheduling}.

\theoremweighteddynamicscheduling*

This theorem details the result of the ``straightforward extension'' of \cref{theorem:unweighted-M=1} for the weighted case, if we assume bounded ratios between the job rewards. In particular, we assume all jobs have rewards within $[1,w]$. The scheduling algorithm guaranteed by the theorem above is at least a factor of $M$ slower than its \MaxIS counterparts. Moreover, the update time of the same algorithm is $\Omega(w)$, while the update time for dynamic weighted interval \MaxIS (see \cref{theorem:weighted-dynamic-M=1}) has no dependence on $w$. The main reason for such behavior of \MaxIS-like algorithms is that they partition time in such a way that each region contains a sparse subproblem, e.g., containing $O(M / \eps)$ jobs, that is easy to solve. However, such regions must have size $\Omega(w)$ in the weighted interval scheduling variant. To see that, consider a long job of reward $w$, with $w$ small non-intersecting jobs of reward $1$ inside it. The optimal scheduling for $M=2$ machines would include all such jobs. However, any partitioning of time that ensures there are $O(M/\eps)$ jobs within each partitioning (akin to the ideas we developed in earlier sections) would discard the long job (removing half the total reward). Thus, intuitively, any algorithm giving better than $2$ approximation would not be able to partition the time-axis as performed in earlier section, and hence all sparse subproblems would have size $\Omega(w)$.

To alleviate this shortcoming, we employ a new \emph{partitioning scheme over machines} to achieve scheduling algorithms that run in $o(M)$ and $o(w)$ time. Instead of a sequential process, we uniformly randomly assign each job to a machine. Then, a job is only allowed to be scheduled on the machine it was assigned to. With these constraints, the interval scheduling problem is equivalent to the \MaxIS problem for each machine given the intervals assigned to it. On the positive side, this results to a scheduling task that computationally can be solved as efficiently as \MaxIS. However, it is unclear what is the approximation loss of this scheduling scheme. Surprisingly, we show that our scheme incurs only the multiplicative factor of $e$ in the approximation loss.

Before we proceed to analyzing the approximation guarantee of this scheme, as a warm-up, we show that compared to \cref{theorem:unweighted-dynamic-scheduling-multiple,theorem:unweighted-LCA-scheduling} this approach yields an even more efficient method for computing unweighted interval scheduling on multiples machines. This efficiency comes at the expense of slightly worsening the approximation guarantee.

\theoremrandomunweighted*
Our proof of \cref{theorem:random-unweighted} is given in \cref{proof:theorem:random-unweighted}. Our main contribution is a black-box result for \emph{weighted} interval scheduling on \emph{multiple} machines, stated as follows.
\theoremweightedrandom*
\begin{proof}
The algorithm begins by immediately assigning each job to one of the machines uniformly at random. Then, we find an optimal solution on each machine with the jobs that were randomly assigned to it, where this subproblem is the \MaxIS problem. Accordingly, this randomized algorithm achieves the same runtime as the oracle for \MaxIS.

Our hope is to show that the union of the optimal solutions for each machine (once we have randomly assigned the jobs), is a high-quality approximation of the globally optimal solution where jobs are not randomly constrained to particular machines. Such a result follows more simply for the proof of \cref{theorem:random-unweighted} in \cref{proof:theorem:random-unweighted}, yet for the weighted case we use a more interesting approach. Instead of directly arguing about the optimal solutions of each \MaxIS problem, we develop a global strategy that respects the random machine constraints and guarantees that each job in $OPT$ has at least a constant probability of being in the final schedule. 

Fix uniformly at random a permutation of the jobs of $OPT$, and consider the jobs of $OPT$ in this order. When we consider a job, we also reveal the machine it is assigned to by $OPT$. Throughout this process, in parallel, we are building an \emph{alternative schedule} as follows. Suppose we are currently considering job $J$ and suppose it has been assigned to machine $P$ by $OPT$. If all the jobs we have scheduled on $P$ so far are either completely contained by $J$ or do not intersect $J$, then we include $J$ in our schedule (deleting all scheduled jobs in $P$ that are contained in $J$). Otherwise, we do not schedule $J$.

Now, we characterize when $J$ is in our schedule at the end of this process. If all jobs completely containing $J$ are assigned to another machine, and all jobs intersecting $J$ that appear earlier in the permutation are assigned to different machines (or are completely contained in $J$), then $J$ will be in our final schedule. As such, a lower bound for the probability $J$ is in our schedule, is the product of
\begin{enumerate}[(1)]
    \item the probability of all jobs containing $J$ being assigned to different machines, and
    \item the probability of all other jobs that intersect $J$ (ignoring jobs $J$ completely contains) and have earlier permutation indices being assigned to different machines.
\end{enumerate}
Suppose there are $C$ jobs that completely contain $J$. Then, no other jobs on those $C$ machines can intersect $J$ as they form a valid schedule. For the remaining $M-1-C$ machines, at most 2 jobs can intersect $J$ that neither completely contain $J$ nor are completely contained within $J$ (both jobs must contain an endpoint of $J$). Thus, the most pessimistic scenario is that there are $C$ machines in $OPT$ containing a job that completely contains $J$ and $M-1-C$ machines in $OPT$ containing two jobs that partially intersect $J$. The probability all $C$ jobs completely containing $J$ are assigned to different machines is $(1-\frac{1}{M})^C$. For the $2(M-1-C)$ jobs that partially overlap with $J$, we take a probability measure over all random permutations. Note that, as the permutation is chosen uniformly randomly, $J$ is equally likely to be at each position of the permutation considering only $J$ and the $2(M-1-C)$ jobs. Moreover, if $J$ is at position $i$, then the probability $J$ is in the final schedule is $(1-\frac{1}{M})^i$. Thus, the probability, all of the intersecting jobs with $J$ are either assigned to different machines or have later permutation positions is $\frac{\sum_{i=0}^{2(M-1-C)} (1-1/M)^i}{2(M-1-C)+1}$.
This gives us a lower-bound where we pessimistically classify machines in the original solution as \emph{containing} machines that have a job completely containing $J$, and \emph{intersecting} machines that have two jobs partially intersecting $J$. For simplicity, we will denote the lower-bound that $C_1$ containing machines do not violate $J$ as $\fcon(C_1) = (1-\frac{1}{M})^{C_1}$ and the lower-bound that $C_2$ intersecting machines do not violate $J$ as $\fint(C_2) = \frac{\sum_{i=0}^{2 C_2} (1-1/M)^i}{2 C_2 +1}$.

Combined, our lower bound that each job is in our schedule is 
$\fcon(C) \times \fint(M-1-C)$,
where $C$ can take integer values in range $0$ to $M-1$. 

\begin{claim}
The quantity $\fcon(C) \times \fint(M-1-C)$ is minimized when $C=M-1$ (i.e., all other machines have one job completely containing this job).
\end{claim}
\begin{proof}
To show this, we show that $\fcon(C) \times \fint(M-1-C) \ge \fcon(M-1) \times \fint(0)$. By factoring out $\fcon(C)$, this is equivalent to showing $\fint(M-1-C) \ge \fcon(M-1-C)$. For simplicity, we set $C'=M-1-C$ and show $\fint(C') - \fcon(C') \ge 0$ for all integer $C'$ from $0$ to $M-1$. Additionally, we define $x=(1-\frac{1}{M})$. As $M>1$, we note that $x \in [\frac{1}{2},1)$. Accordingly:

\begin{align*}
    & \fint(C') - \fcon(C') \\
    = & \frac{\sum_{i=0}^{2 C'} (1-1/M)^i}{2 C' +1} - (1-1/M)^{C'} \\
    = & \frac{\sum_{i=0}^{2 C'} x^i}{2 C' +1} - x^{C'} \\
    = & \frac{\sum_{i=0}^{C' - 1} (x^{C'-i} + x^{C'+i} - 2 x^{C'})}{2 C' +1} \\
    = & \frac{\sum_{i=0}^{C' - 1} (x^{C' - i} \times (1 + x^{2i} - 2 x^{i}))}{2 C' +1} \\
    = & \frac{\sum_{i=0}^{C' - 1} (x^{C' - i} \times (x^{i} - 1)^2)}{2 C' +1} \\
    \ge & 0
\end{align*}

The last step is obtained because each summand is non-negative. This shows that $\fint(C') \ge \fcon(C')$ for all valid integer $C'$ and thus $\fcon(C) \times \fint(M-1-C)$ is minimized when $C=M-1$.

\end{proof}

Thus our lower bound of a job being in the resulting solution is always at least $\fcon(M-1) = (1-\frac{1}{M})^{M-1} \ge \frac{1}{e}$.

With this, we show that our generative process results in a schedule on average that has weight at least $\frac{|OPT|}{e}$. This implies a $\alpha$-approximate \MaxIS algorithm yields an $e \alpha$-approximation
\end{proof}

As such, we explore the relationship between partitioning over time and machines to solve the interval scheduling problem. To achieve a $(1+\eps)$-approximations for unweighted and $(\frac{e}{e-1}+\eps)$-approximations for weighted scheduling, we simultaneously partition over time and machines at the expense of slower algorithms. However, if we tolerate $(2-1/M+ \eps)$-approximations for unweighted scheduling or $(e + \eps)$-approximations for weighted scheduling,
 we randomly partition over machines \emph{then} time to achieve comparable efficiency to the \MaxIS problem.

\subsection{Proof of \cref{theorem:unweighted-dynamic-scheduling-multiple}}
\label{proof:theorem:unweighted-dynamic-scheduling-multiple}
We maintain a modified version of \cref{invariant:unweighted-MIS}, where the algorithm maintains a set of borders such that an optimal solution for within each two consecutive borders is of size between $\nicefrac{M}{\eps}$ and $\nicefrac{2M}{\eps}+M$ jobs. Direct modification of \cref{lemma:invariant-value} shows that this maintains a $(1+\eps)$-approximation (the size of solutions between consecutive borders is a factor of $M$ larger than in \cref{theorem:unweighted-M=1} because $M$ jobs may intersect any border).

As a starting point, we consider the classic greedy algorithm for unweighted interval scheduling on multiple machines \cite{tardos2005algorithm}, that we recall next:
\begin{itemize}
    \item Among jobs that start after the earliest time any machine is free, find the one with the earliest ending time.
    \item Then, among machines that can take the job, schedule the job to the machine that becomes free at the latest time.
\end{itemize}
This solution can easily be simulated in $O(|OPT|\log(n))$ time by a method similar to \ProbeBased. In our dynamic version, we handle insertions and deletions analogously to as in \cref{theorem:unweighted-M=1}.
More specifically, when a job is added inside a region, we recompute an answer for the region in $O(\frac{M \log(n)}{\eps})$ time. If the solution becomes too large, we add a border after the $\frac{M}{\eps}$-th ending point of a job in the solution (this will invalidate at most $M$ jobs, leaving the left half with a $[\nicefrac{M}{\eps},\nicefrac{M}{\eps}+M]$ size solution and the right half with a $[\nicefrac{M}{\eps}+1, \nicefrac{M}{\eps}+M+1]$ size solution). If deleting a job makes the recomputed solution too small, we combine with an adjacent region (and if the region is now too large, we add a new border to split the region like above).

With essentially the same approach as \cref{theorem:unweighted-M=1}, we obtain $O(\frac{M \log(n)}{\eps})$ updates and $O(\log(n))$ queries in worst-case time.

\subsection{Proof of \cref{theorem:unweighted-LCA-scheduling}}
\label{proof:theorem:unweighted-LCA-scheduling}
First, we modify the \emph{successor oracle} for this result. Consider an instance with two machines and two jobs corresponding to time ranges $[1,4]$ and $[2,3]$. No successor oracle probe will ever return the first job because a successor oracle prove will never return a job completely contained in another job. Thus the original successor oracle is not strong enough to determine any particular constant-factor approximation to scheduling even with infinite probes. To remedy this, we modify the successor oracle such that it ignores a set of jobs given with the probe (it is not a concern that this set will be very large, as any probe-efficient algorithm will not know many jobs to specify for the set) which enables us to simulate a subroutine analogous to \ProbeBased.

With this new successor oracle, our algorithm and analysis is almost identical to \cref{alg:alg-local-approx} proven in \cref{theorem:unweighted-M=1-local}. Our key difference is that we now set the thresholds for drawing borders to when $OPT(Q_L)$ and $OPT(Q_R)$ are larger than $\nicefrac{M}{\eps}$ instead of $\nicefrac{1}{\eps}$. With this, we are maintaining the modified version of \cref{invariant:unweighted-MIS} from \cref{theorem:unweighted-dynamic-scheduling-multiple} that is shown to result in a $(1+\eps)$-approximation. More concretely, to simulate this process we define a function $f(Q,first\_empties)$ analogous to that of $f(Q,first\_empty)$ from \cref{alg:alg-local-approx}. The primary differences are the aforementioned factor of $M$ increase of the threshold for drawing a border, our simulation of \ProbeBased is thus a factor of $M$ slower, we have $M$ possible $I_{mid}$, and we keep track of and return the times for all $M$ machines (hence $first\_empties$ instead of $first\_empty$). 

With essentially the same approach as \cref{theorem:unweighted-M=1-local}, we obtain a local computation algorithm for $(1+\eps)$-approximate unweighted interval scheduling on $M$ machines using $O(\frac{M\log(N)}{\eps})$.

\subsection{Proof of \cref{theorem:weighted-dynamic-scheduling}}
\label{proof:theorem:weighted-dynamic-scheduling}
We outline an alternative approach to a dynamic algorithm for weighted independent set of intervals based on \cref{section:dynamic-unit}. While a stronger result is presented in \cref{section:dynamic-weighted}, that approach does not easily lend itself well to repeatedly calculating \MaxIS. We instead build off the simpler result from \cref{section:dynamic-unit}. 

We maintain a modified version of \cref{invariant:unweighted-MIS}, where the reward of the solution we calculate within consecutive borders is in range $[\nicefrac{Mw}{\eps},\nicefrac{8Mw}{\eps}+2Mw]$. We want to repeatedly calculate a $(1+\eps)$-approximation of \MaxIS within regions and use a similar but different approach to \cref{lemma:sparse-sol-approx}. In contrast to the setting of \cref{lemma:sparse-sol-approx}, our invariant bounds the total weight within consecutive borders as opposed to the number of jobs in the optimal solution within consecutive borders. Consider a dynamic programming problem where our state is the total weight of jobs and the corresponding answer is the shortest prefix that can obtain jobs of this total weight. It is simpler for us if all weights are integers and there are not many weights. We round all weights down to powers of $(1+\eps)$, which will not affect our approximation by more than a factor of $(1+\eps)$. Then, we scale all weights by $\nicefrac{1}{\eps}$. Each job now has weight at least $\nicefrac{1}{\eps}$, so rounding down to the nearest integer is at most $\eps$ fraction of the weight and the remaining optimal solutions is still an $(1+O(\eps))$-approximation. Now, we optimally calculate the \MaxIS within each region given the rounding. Let $D$ be the number of distinct weights. The dynamic programming problem we mentioned can be solved in $O(\log(n) \cdot |OPT| \cdot D)$ time as there are at most $|OPT|$ states we can reach, there are $D$ possible transitions (trying the job with some given weight that starts after the current prefix and ends earliest), and each transition uses a $O(\log(n))$ query to a balanced binary search tree. Due to our invariant and scaling weights, the sum of $|OPT|$ as we calculate \MaxIS $M$ times is at most $O(\frac{Mw}{\eps^2})$. By rounding the weights down to powers of $(1+\eps)$, $D = O(\frac{\log(w)}{\eps})$. Thus, we recompute the answer for a region in $O(\frac{Mw \log(n) \log(w)}{\eps^3})$. 

Now, we handle insertions and deletions similarly to \cref{theorem:unweighted-M=1}. This maintains a $((\frac{M^M}{M^M-(M-1)^M})(1 + \eps))$-approximation, which is also a $4$-approximation.
This means the solution the algorithm generates for any region is at least $\frac{1}{4}$ the optimal solution for that region. When we insert/delete intervals in a region, we recompute the answer for the region. If the total weight of the region becomes too small, we repeatedly combine with adjacent regions until it is not too small. At most four combinations must occur, as then the union of the solutions we had found is at last a factor of $4$ larger than the minimum solution size for a region, so our 4-approximation must find it. If we add a job and the region solution becomes too large, we note that the true solution size is at most $4(\frac{8Mw}{\eps}+2Mw)$. Whenever a region's solution is too large, we split at the smallest prefix that contains intervals of total weight $\frac{4Mw}{\eps}$. The left region will have a solution of size $\ge \frac{4Mw}{\eps}$ and the right region will have a solution of size $\ge (\frac{8Mw}{\eps}+2Mw)-(\frac{4Mw}{\eps}-2Mw) = \frac{4Mw}{\eps}$. Thus, our 4-approximation will find a solution of size at least $\frac{Mw}{\eps}$ for both and we will never classify either as too small. As we separate at least $\frac{4Mw}{\eps}$ of weight with every split, only $O(1)$ splits will occur. With this, we achieve an algorithm with $O(\frac{Mw \log(n) \log(w)}{\eps^3})$ update and $O(\log(n))$ query time worst-case.

\subsection{Proof of \cref{theorem:random-unweighted}}
\label{proof:theorem:random-unweighted}
The algorithm begins by assigning each job to one of the machines uniformly at random. Then, finding an optimal solution on each machine is the \MaxIS problem.
Our proof technique is to simultaneously simulate the classic greedy \MaxIS algorithm and the realization of each job's assignment for a single machine. We show that the expected \MaxIS of jobs assigned to a machine is at least $\frac{|OPT|}{2M-1}$. 

Consider the set of jobs in an optimal solution $OPT$, and ignore all others. In the classic greedy \MaxIS algorithm, we consider jobs in an increasing order of their ending time and use the job if it does not intersect any previously selected jobs. At a high-level, we will simulate this algorithm on a particular machine, realizing whether or not a job was assigned to this machine only as we need to. In particular, assume we have a set of jobs $OPT$ that are a valid scheduling on $M$ machines and all start after the ending points of any jobs we have previously selected. We consider this set in increasing order of ending time. When we consider a job $I$, we realize its assignment. If $I$ is not assigned to the current machine (probability $1-1/M$), we cannot use it. If $I$ is assigned to the current machine (probability $1/M$), we use the job and delete all jobs in $OPT$ that intersect it. Note that all other jobs in $OPT$ have an ending time that is at least the ending time of $I$ (because we have not yet considered them). Thus, to intersect $I$, they must start before $I$ ends. This implies that all the jobs we delete must contain the ending point of $I$. Since $OPT$ is a valid schedule for $M$ machines (and no schedule on $M$ machines can have $>M$ jobs containing a point), we only need to delete at most $M-1$ other jobs. In either situation, the invariant on $OPT$ is maintained afterwards.

Thus, we write the following recurrence $f(X)$ to denote a lower bound on the expected size of the \MaxIS given $|OPT|=X$:
\[
    f(X) \ge (1 - 1/M) f(X-1) + \nicefrac{1}{M} (f(X-M) + 1).
\]
For simplicity of notation, we assume that $f(X) = 0$ for $X \le 0$.

\begin{lemma}
It holds that $f(X) \ge \frac{X}{2M-1}$.
\end{lemma}
\begin{proof}
First, we show the claim when $X \le M$.
We have the following chain of inequalities:
\begin{align*}
    & \rb{1-\frac{1}{M}}f(X-1) + \frac{1}{M} (f(X-M) + 1) \\
    \ge & \rb{1 - \frac{1}{M}} \frac{X-1}{2M-1} + \frac{1}{M} (0 + 1) \\
    = & \frac{X-1}{2M-1} - \frac{X-1}{(2M-1)M} + \frac{1}{M}
    \\
    = & \frac{M(X-1) - (X-1) + (2M-1)}{M (2M-1)} \\
    = & \frac{Mx + M - x}{M (2M-1)} \\
    \ge & \frac{Mx}{M (2M-1)} \\
    = & \frac{X}{2M-1}.
\end{align*}

Next, we show the claim when $X > M$:
\begin{align*}
    & \rb{1-\frac{1}{M}}f(X-1) + \frac{1}{M} (f(X-M) + 1) \\
    \ge & \rb{1 - \frac{1}{M}}\frac{X-1}{2M-1} + \frac{1}{M}\rb{\frac{X-M}{2M-1}+1} \\
    = & \frac{Mx}{M (2M-1)} \\
    = & \frac{X}{2M-1}.
\end{align*}
\end{proof}

Thus, we have that $f(|OPT|) \ge \frac{|OPT|}{2M-1}$. As all machines are identical, the expected value of the schedule is the sum of their expected \MaxIS. Thus, the expected optimal schedule has size $\frac{M |OPT|}{2M-1}$. Using an $\alpha$-approximation for each of these \MaxIS subproblems yields a $(2-1/M)\alpha$-approximation, as advertised.

Note that this bound is tight as $n$ approaches infinity. Consider an instance with $n$ jobs, where job $i$ starts at time $i$ and ends at time $i+M$. If we simulate the classic greedy algorithm for \MaxIS on a machine, it will see $M$ jobs in expectation until it sees one that is assigned to it (expectation of a Bernoulli random variable). To use this interval, the $M-1$ jobs after it cannot be used (they all intersect). Thus, for every job in the solution, in expectation the machine needed to throw away $2M-2$ other jobs and thus as $n$ approaches infinity the expected schedule size approaches $\frac{M}{2M-1}$.

\section*{Acknowledgements}
We thank Benjamin Qi (MIT) for helpful discussions. S.~Compton was supported in part by the National Defense Science \& Engineering Graduate (NDSEG) Fellowship Program. S.~Mitrovi\' c was supported by the Swiss NSF grant No.~P400P2\_191122/1, NSF award CCF-1733808, and FinTech@CSAIL. R.~Rubinfeld was supported by the NSF TRIPODS program (awards CCF-1740751 and
DMS-2022448), NSF award CCF-2006664, and FinTech@CSAIL.
\bibliography{main}

\begin{thebibliography}{10}

\bibitem{agarwal1998label}
Pankaj~K Agarwal and Marc~J Van~Kreveld.
\newblock {\em Label placement by maximum independent set in rectangles},
  volume 1998.
\newblock Utrecht University: Information and Computing Sciences, 1998.

\bibitem{alon2012space}
Noga Alon, Ronitt Rubinfeld, Shai Vardi, and Ning Xie.
\newblock Space-efficient local computation algorithms.
\newblock In {\em Proceedings of the twenty-third annual ACM-SIAM symposium on
  Discrete Algorithms}, pages 1132--1139. Society for Industrial and Applied
  Mathematics, 2012.

\bibitem{arkin1987scheduling}
Esther~M Arkin and Ellen~B Silverberg.
\newblock Scheduling jobs with fixed start and end times.
\newblock {\em Discrete Applied Mathematics}, 18(1):1--8, 1987.

\bibitem{bar2001approximating}
Amotz Bar-Noy, Sudipto Guha, Joseph Naor, and Baruch Schieber.
\newblock Approximating the throughput of multiple machines in real-time
  scheduling.
\newblock {\em SIAM Journal on Computing}, 31(2):331--352, 2001.

\bibitem{bhore2020dynamic}
Sujoy Bhore, Jean Cardinal, John Iacono, and Grigorios Koumoutsos.
\newblock Dynamic geometric independent set.
\newblock {\em arXiv preprint arXiv:2007.08643}, 2020.

\bibitem{buttazzo2012limited}
Giorgio~C Buttazzo, Marko Bertogna, and Gang Yao.
\newblock Limited preemptive scheduling for real-time systems. a survey.
\newblock {\em IEEE Transactions on Industrial Informatics}, 9(1):3--15, 2012.

\bibitem{cardinal2021worst}
Jean Cardinal, John Iacono, and Grigorios Koumoutsos.
\newblock Worst-case efficient dynamic geometric independent set.
\newblock In {\em 29th Annual European Symposium on Algorithms (ESA 2021)},
  volume 204, page~25. Schloss Dagstuhl--Leibniz-Zentrum f $\{$$\backslash$"
  u$\}$ r Informatik, 2021.

\bibitem{correa2005single}
Jos{\'e}~R Correa and Andreas~S Schulz.
\newblock Single-machine scheduling with precedence constraints.
\newblock {\em Mathematics of Operations Research}, 30(4):1005--1021, 2005.

\bibitem{frank1976some}
A~FRANK.
\newblock Some polynomial algorithms for certain graphs and hypergraphs.
\newblock In {\em Proceedings of the 5th British Combinatorial Conference,
  1975}. Utilitas Mathematica, 1975.

\bibitem{gavruskin2014dynamic}
Alexander Gavruskin, Bakhadyr Khoussainov, Mikhail Kokho, and Jiamou Liu.
\newblock Dynamic interval scheduling for multiple machines.
\newblock In {\em International Symposium on Algorithms and Computation}, pages
  235--246. Springer, 2014.

\bibitem{gavruskin2015dynamic_monotonic}
Alexander Gavruskin, Bakhadyr Khoussainov, Mikhail Kokho, and Jiamou Liu.
\newblock Dynamic algorithms for monotonic interval scheduling problem.
\newblock {\em Theoretical Computer Science}, 562:227--242, 2015.

\bibitem{gawrychowski2022sublinear}
Pawe{\l} Gawrychowski and Karol Pokorski.
\newblock Sublinear dynamic interval scheduling (on one or multiple machines).
\newblock {\em arXiv preprint arXiv:2203.14310}, 2022.

\bibitem{henzinger2020dynamic}
Monika Henzinger, Stefan Neumann, and Andreas Wiese.
\newblock Dynamic approximate maximum independent set of intervals, hypercubes
  and hyperrectangles.
\newblock In {\em 36th International Symposium on Computational Geometry (SoCG
  2020)}. Schloss Dagstuhl-Leibniz-Zentrum f{\"u}r Informatik, 2020.

\bibitem{hochbaum1985approximation}
Dorit~S Hochbaum and Wolfgang Maass.
\newblock Approximation schemes for covering and packing problems in image
  processing and vlsi.
\newblock {\em Journal of the ACM (JACM)}, 32(1):130--136, 1985.

\bibitem{kolen2007interval}
Antoon~WJ Kolen, Jan~Karel Lenstra, Christos~H Papadimitriou, and Frits~CR
  Spieksma.
\newblock Interval scheduling: A survey.
\newblock {\em Naval Research Logistics (NRL)}, 54(5):530--543, 2007.

\bibitem{lenstra1978complexity}
Jan~Karel Lenstra and AHG Rinnooy~Kan.
\newblock Complexity of scheduling under precedence constraints.
\newblock {\em Operations Research}, 26(1):22--35, 1978.

\bibitem{levey20191}
Elaine Levey and Thomas Rothvoss.
\newblock A (1+ epsilon)-approximation for makespan scheduling with precedence
  constraints using lp hierarchies.
\newblock {\em SIAM Journal on Computing}, pages STOC16--201, 2019.

\bibitem{mingozzi1999set}
Aristide Mingozzi, Marco~A Boschetti, Salvatore Ricciardelli, and Lucio Bianco.
\newblock A set partitioning approach to the crew scheduling problem.
\newblock {\em Operations Research}, 47(6):873--888, 1999.

\bibitem{pinedo2012scheduling}
Michael Pinedo.
\newblock {\em Scheduling}, volume~29.
\newblock Springer, 2012.

\bibitem{robert2008non}
Julien Robert and Nicolas Schabanel.
\newblock Non-clairvoyant scheduling with precedence constraints.
\newblock In {\em Proceedings of the nineteenth annual ACM-SIAM symposium on
  Discrete algorithms}, pages 491--500, 2008.

\bibitem{rubinfeld2011fast}
Ronitt Rubinfeld, Gil Tamir, Shai Vardi, and Ning Xie.
\newblock Fast local computation algorithms.
\newblock {\em arXiv preprint arXiv:1104.1377}, 2011.

\bibitem{salot2013survey}
Pinal Salot.
\newblock A survey of various scheduling algorithm in cloud computing
  environment.
\newblock {\em International Journal of Research in Engineering and
  Technology}, 2(2):131--135, 2013.

\bibitem{sharma2010survey}
Raksha Sharma, Vishnu~Kant Soni, Manoj~Kumar Mishra, and Prachet Bhuyan.
\newblock A survey of job scheduling and resource management in grid computing.
\newblock {\em world academy of science, engineering and technology},
  64:461--466, 2010.

\bibitem{skutella2005stochastic}
Martin Skutella and Marc Uetz.
\newblock Stochastic machine scheduling with precedence constraints.
\newblock {\em SIAM Journal on Computing}, 34(4):788--802, 2005.

\bibitem{tardos2005algorithm}
Eva Tardos and Jon Kleinberg.
\newblock Algorithm design, 2005.

\bibitem{verweij1999optimisation}
Bram Verweij and Karen Aardal.
\newblock An optimisation algorithm for maximum independent set with
  applications in map labelling.
\newblock In {\em European Symposium on Algorithms}, pages 426--437. Springer,
  1999.

\end{thebibliography}

\end{document}